\documentclass[12pt, titlepage]{article}
\usepackage{amsmath,amsthm,amssymb,endnotes,lamacs,pdfsync}
\usepackage[comma]{natbib}
\usepackage{color}
\renewcommand{\cite}{\citep}

\usepackage[scaled=.92]{helvet}
\usepackage[T1]{fontenc}
\usepackage{textcomp}
\usepackage{multirow}
\newcommand{\dimp}{\Leftrightarrow}
\newcommand{\AX}{\ifmmode\mathrm{AX}\else\textbf{AX}\fi}

\newcommand{\A}{{\cal A}}
\newcommand{\I}{{\cal I}}
\renewcommand{\P}{{\cal P}}
\newcommand{\U}{{\cal U}}
\newcommand{\V}{{\cal V}}
\newcommand{\T}{{\cal T}}

\newcommand{\thm}{\begin{theorem}}
\newcommand{\ethm}{\end{theorem}}
\newtheoremstyle{axiom}%
{\topsep}{\topsep}{}{}{\bfseries\upshape}{.}{ }{}
\theoremstyle{axiom}
\newtheorem{axiom}{A\hspace{-3pt}}

\newtheorem*{atwo}{$\AthreeprAX$}
\newtheorem*{arch}{Arch}

\newcommand{\intr}{\cap}
\newcommand{\union}{\cup}
\newcommand{\bbox}{\vrule height7pt width4pt depth1pt}
\newcommand{\<}{\langle}\renewcommand{\>}{\rangle}
\newcommand{\rimp}{\Rightarrow}
\newcommand{\true}{\mathit{true}}
\newcommand{\false}{\mathit{false}}
\newcommand{\btrue}{\mathbf{true}}
\newcommand{\bfalse}{\mathbf{false}}

\newcommand{\ifthnels}{\textbf{if} \ldots \textbf{then} \ldots
\textbf{else}}
\newcommand{\ift}[3]
{\mbox{\textbf{if ${#1}$ then ${#2}$ else ${#3}$}}}
\newcommand{\rhoso}{\rho_{SO}}

\newtheorem{prop}{Proposition}
\newtheorem{Example}{Example}
\newcommand{\xam}{\begin{Example}}
\newcommand{\exam}{\bbox\end{Example}}

\newcommand{\Aone}{\ifmmode\mathbf{A1}\else\textbf{A1}\fi}
\newcommand{\AthreeAX}{\mathbf{A2}}
\newcommand{\AthreeprAX}{\mathbf{A2'}}

\newcommand{\AthreepsAX}{\mathbf{A2^*}}
\newcommand{\AthreeprpAX}{\mathbf{A2^\dag}}
\newcommand{\Afour}{\ifmmode\mathbf{A3}\else\textbf{A3}\fi}
\newcommand{\Arch}{\ifmmode\mathbf{Arch}\else\textbf{Arch}\fi}
\newcommand{\Afive}{\ifmmode\mathbf{A4}\else\textbf{A4}\fi}
\newcommand{\Asix}{\ifmmode\mathbf{A5}\else\textbf{A5}\fi}
\newcommand{\Aseven}{\ifmmode\mathbf{A6}\else\textbf{A6}\fi}

\newcommand{\ecsa}{extended statewise cancellation}
\newcommand{\csa}{statewise cancellation}
\newcommand{\Csa}{Statewise cancellation}

\newcommand{\ca}{cancellation}
\newcommand{\ecm}{extended mixture cancellation}

\DeclareMathOperator{\Atoms}{At}

\newcommand{\AtomsAX}{\Atoms_{\T}}

\newcommand{\commentout}[1]{}
\newcommand{\fullv}[1]{#1}
\newcommand{\shortv}[1]{\commentout{#1}}

\newcommand{\EXAX}{\mathit{EX_\mathrm{AX}}}
\newcommand{\EXAXp}{\mathit{EX_\mathrm{AX}^+}}
\newcommand{\CC}{\mathit{CC}}

\newenvironment{thmlist}{\begin{list}{\arabic{enumi}.}
{\usecounter{enumi}
\setlength{\topsep}{-\parskip}\addtolength{\topsep}{3pt}
\setlength{\itemindent}{0pt}
\setlength{\leftmargin}{3pc}
\setlength{\labelwidth}{1pc}}}{\end{list}}



\begin{document}
\title{Constructive Decision Theory%
\renewcommand{\thefootnote}{}\footnotetext{\hskip-2\parindent\kern3pt
Supported in part by NSF under grants
CTC-0208535, ITR-0325453, IIS-0534064, IIS-1703846, and IIS-1718108, by
ONR under grants N00014-00-1-03-41 and N00014-01-10-511, 
by ARO under grant W911NF-17-1-0592, 
by a USDoD Multidisciplinary University Research Initiative (MURI) program
administered by the ONR under grant N00014-01-1-0795, and
by a MURI program administered by the ARO under grant W911NF-19-1-0217.
An extended
abstract of this paper \citep{BEH06} appeared in the {\it Tenth International
Conference on Principles of Knowledge Representation and Reasoning}.
An early version of this paper was written for and first presented at
a conference on decision theory organized by Karl Vind in May 2004,
shortly before his death.  As a scholar and a colleague, Karl has left
us all much to be grateful for.  We are particularly grateful to him
for pushing us to deliver a paper for this conference, and then to
generously comment on it despite his ill health. We are also grateful
to several individuals for comments, to the associate editor and
reviewers of the paper, 
and to the many seminar audiences 
who have listened to and commented on this paper. 
In particular, we thank Jin-Yi Cai for suggesting the proof of Proposition~\ref{prop:cancellationbound}.
}}  \author{Lawrence Blume$^{a,c}$, David Easley$^{a}$ and Joseph
  Y. Halpern$^b$\\~\\
$^a$ Dept.~of Economics and Dept.~of Information Science, Cornell University\\
$^b$ Dept.~of Computer Science, Cornell University\\
$^c$ Institute for Advanced Studies, Vienna}
\date{\today}
\maketitle

\thispagestyle{empty}
\noindent\textbf{Abstract:}
\noindent In most contemporary approaches to decision making under uncertainty, a decision
problem is described by a set of states and set of outcomes, and a rich set
of \emph{acts}, which are functions from states to outcomes over which the
decision maker (DM) has preferences. Many interesting decision problems,
however, do not come with a state
space and an outcome space.  Indeed, in complex problems it is often far
from clear what the state and outcome spaces would be.  We present an
alternative foundation for decision making, in which the primitive objects
of choice are \emph{syntactic programs}. A representation theorem is
proved in the spirit of standard representation theorems,
showing that if the DM's preference relation on objects of choice satisfies
appropriate axioms, then there exist a set $S$ of states, a set $O$ of
outcomes, a way of interpreting the objects of choice as functions from $S$ to $O$, a
probability on $S$, and a utility function on $O$, such that the DM
prefers choice $a$ to choice $b$ if and only if the expected
utility of $a$ is higher than that of $b$.  Thus, the state space and
outcome space are subjective, just like the probability and utility;
they are not part of the description of the problem.  In principle, a
modeler can test for SEU behavior without having access to states or
outcomes. We illustrate the power of our approach by showing that it
can capture decision makers who are subject to framing effects and those who are subject to failures of extensionality.

\noindent\textbf{Correspondent:}\newline
Professor David Easley\newline
Department of Economics\newline
Uris Hall\newline
Cornell University\newline
Ithaca NY 14850
\newpage
\setcounter{page}{1}

\noindent\emph{In memoriam} Karl Vind.
\section{Introduction}
Models of decision making under uncertainty typically begin with
states of the world, outcomes, acts, which are functions mapping
states to outcomes, and preferences over acts.  This Savage
(\citeyear{Savage}) presentation is convenient for the analysis of
choice behavior, in particular for understanding how choices vary with
those things the analyst interprets as variations in preferences such
as tastes and beliefs.  But this representation is often not close to the
way in which choice problems appear in the world, and how
decision makers (DMs) talk and reason about their decisions.  

Decision problems are expressed not in the language of states, acts, and
orders, but instead in some natural language in which the basic
objects are statements such as `the broker recommends a purchase of
IBM' or `the broker recommends a purchase of Alphabet' and the objects
of choice are `use my cash to buy IBM', `use my cash to buy Alphabet',
or `do nothing'. The correct mapping of these statements into states
and outcomes may seem obvious to the analyst, but what guarantees that
the decision maker acts as if he uses any such mapping?  The move from
the language of decisions in the world to the language of decisions in
theory is made by the analyst; it is subjective, it usually proceeds
without discussion, and yet it determines the counterfactual claims
that empirical economics uses to make causal assertions.  Furthermore,
much of the thrust of behavioral decision theory alludes to or even
requires some unpacking of the DM's reasoning.  This is difficult to
do when the language of the model is so different from the language of
the DM.    
In fact, \citet{Grabiszewski16} shows that the existence of an SEU
representation for preferences in a decision problem has little
empirical content.  If the preference relation in a finite Savage
presentation of a decision problem is complete, transitive, and
satisfies a monotonicity requirement, then there exists some Savage
representation into which the problem can be embedded in a
preference-preserving way, and which has an SEU representation.  This
is simply to say that the empirical content of SEU preferences derives from the
meaning of the states and outcomes.  Starting with a natural-language
description of a decision problem, we derive states, outcomes, and
Savage acts in terms of the choices a DM can make and the things she
can observe about the world, so that the Savage presentation is in
some sense close to the DM's understanding of the problem. 

Our approach is to model choice objects as syntactic statements of the
form `if $t_1$ then choose $a_1$, else if $t_2$ choose $a_2$, \ldots',
which we view as programs in a programming languge.
The $a_i$ are actions (not `acts') available to the DM.  We call the
$t_i$ `tests'; they are propositions about the
world.  The programming
language we focus on in this paper is very simple---we use it just to
illustrate our ideas.  Critically, it includes tests (in the context
of \ifthnels\ statements).  These tests involve syntactic descriptions
of the events in the world, and allow us to distinguish events from
(syntactic) descriptions of events.  In particular, there can be two
different descriptions that, intuitively, describe the same event from
the point of view of the modeler but may describe different events
from the point of view of the decision maker.  We assume that a DM has
a preference relation on these programs,\footnote{Or, more precisely, as 
a referee pointed out, on the programs to which these syntactic descriptions 
refer.} and from this derives a Savage
representation: states, outcomes, acts, a corresponding preference
relation on acts, and an expected utility representation for these
preferences.

While bringing the description of choice problems closer to the way
they appear in the world to DMs is a natural move in its own right, it
also has several concrete advantages.  First, empirical decision
theory provides joint tests of the expected utility hypothesis and the 
analyst's interpretation as a Savage model.  It will become clear that
the Savage representation of any nontrivial decision problem is not
unique, that many different Savage representations could be deployed.
Separating the DM's understanding of the problem from the analyst's
interpretation has obvious advantages, among them the ability to
determine which properties of choice are common to all Savage
representations and which are representation-specific.  This is
particularly important for research that proposes to `test
rationality'.  It would be unfortunate if rationality meant `agreement
with the experimenter's view of the world'.
We prefer to
create a more flexible framework for rational choice and then add
rejectable axioms to generate sharper predictions.  

%
Second, in our approach, framing anomalies can often be understood
as a conflict between the modeler's and the DM's representations.
Although
we do not pursue it here, our approach allows us to model the effects
of changing the DM's view of the world by considering more (or a
different set of) tests.   Like \citet{tverskyetal94}, but unlike
\citet{AE07}, our approach to framing begins with a natural-language
description of the world.  Unlike Tversky and Koehler, we construct
a state space from the natural-language tests rather than taking
it as externally given. 

Third, our approach provides a natural way to capture resource-bounded
reasoning---and even incorrect reasoning---about the world in an
expected-utility framework.  Fourth, our framework can model the interaction of
DMs who describe the world in different ways.  Finally,
the modelling of many behavioral effects has often required the
application of decision theories that are somewhat difficult to
manipulate, for instance, requiring non-additive beliefs, menu-choice
constructions, and the like.  These models are often difficult to
compare with each other, and they can be very difficult to implement.
The basic expected utility framework is durable because although it is
not great, it is often not bad; and because additive separability
makes for representations that are easy to implement empirically.
Our approach extends the domain of additive separability some way into
regimes that have normatively been labelled irrational.  The
natural-language modeling approach provides additional hooks for
theorizing about anomalous behavior within an expected-utility framework.  

We began this project because we believed that for many decision
problems, the assumption that there is a natural state space that
would describe the uncertainty, and that DMs could (implicitly)
articulate this in their reasoning, is ridiculous.  Equally absurd is
the assumption that in such circumstances individuals could articulate
a complete preference relation over all alternatives.
We therefore choose to allow for incomplete preferences.  Of course,
with incompleteness comes a loss of uniqueness of the representation.
This does not concern us because, as we have observed, the
the way the DM models the original natural-language problem in terms
of states and outcomes is not unique either.  On
the other 
hand, we show that there always is a \emph{canonical} representation,
and the description of a state will be different depending on whether
preferences are incomplete or not, and if so, what is missing.

Many papers in the literature raise issues with the state-space
approach of Savage, or derive a subjective state 
space. \citet{Machina03} surveys the standard approach and illustrates many difficulties with the theory and with its uses. These
difficulties include the ubiquitous ambiguity over whether the theory
is meant to be descriptive or normative, whether states are exogenous
or constructed by the DM, whether states are external to the DM, and
whether they are measurable or not. \citet{Kreps92} and
\citet{DLR01} use a menu choice model to deal with unforeseen
contingencies---an inability of the DM to list all possible states of
the world. They derive a subjective state space that represents
possible preference relations over elements of the menu chosen by the
DM. \citet{Ghir01} takes an alternative approach to unforeseen
contingencies and models acts as correspondences from a state space to
outcomes.  \citet{GS04} and \citet{Karni06} raise objections
to the state space that are similar to ours, and develop decision
theories without a state space. Both papers derive subjective
probabilities directly on outcomes.  \citet{Ahn07} also develops
a theory without a state space; in his theory, the DM chooses over
sets of lotteries over consequences.  \citet{AE07} allow for the
possibility that there may be different descriptions of a particular
event, and use this possibility to capture framing.  For them, a
`description' is a partition of the state space.  They provide an
axiomatic foundation for decision making in this framework, built on
Tversky and Koehler's (\citeyear{tverskyetal94}) notion of
\emph{support theory}.
(As we shall see, our approach is also quite compatible with support theory.)
\citet{Grabiszewski16} asks whether a decision maker's preferences over acts
mapping a given state space to a known outcome space can be SEU-rationalized 
with an alternative state space even if the original preferences were inconsistent
with SEU. His analysis is related to our analysis in Section 4.3 with an objective outcome space. 
The primary difference is that we do not begin with states and outcomes;
rather our DM has preferences over objects of choice (programs) described
in his own language.  
\citet{BV18} fix a state space and outcome space, but allow the DM to have a
subjective interpretation of what they call \emph{feasible acts},
which can be viewed as syntactic objects; each feasible act is
interpreted as a function from states to outcomes. 
Finally, \citet{Lip99} also considers language and uses a
subjective state space that, intuitively, may include `impossible
possible worlds', where the standard axioms of logic may not
hold. He shows how his approach can capture framing problems, among
other things.   
Although there is clearly some overlap in intuitions, the technical
details of our approach are significantly different from those
mentioned above.  Perhaps the closest is the work of Lipman; we
discuss its relation to our work in Section~\ref{sec:cancallationchoices}. 

The rest of this paper is organized as follows. In the next
section, we introduce the syntactic programs that we take as
our objects of choice, discuss several interpretations of the model,
and show how syntactic programs can be interpreted as Savage
acts. This section also includes several examples illustrating the
power of our approach.  In Section 3, we present our assumptions on
preferences.  Because the set of programs does not have a
mixture-space structure, we replace independence assumptions
with Krantz et al.'s (\citeyear{KLST71}) \emph{cancellation axiom}.
In Section 4 we present our representation theorems for decision
problems with subjective outcomes and those with objective
outcomes. Section 5 discusses how our framework can model boundedly
rational reasoning.  In Section 6 we discuss how updating works for
new information about the external world as well as for new
information about preferences. Our goal in this paper is to introduce
our approach and to relate it to the classical Savage approach.  We
conclude in Section 7 with a discussion of the benefits of the
approach and further suggestions for how it can be applied.

\section{Describing Decision Problems}
As usual, we assume that the agent chooses among acts, but as we said
in the introduction, for us, the acts are programs in a simple
programming language.  So we begin by describing the language of tests, and then use this
language to construct programs, our syntactic objects of choice.
We then discuss how the language of tests is used to describe the DMs theory of the world.
We conclude the section by giving a several examples where this approach
can capture behavior that is difficult to explain using more standard approaches.

\subsection{Languages for tests and choices}
A \emph{primitive test} is a statement about the world that is either true
or false, such as `the economy will be strong next
year' and `the moon is in the seventh house'. We assume a finite set
$T_0$ of primitive tests.  The set $T$ of tests is constructed by
closing the set of primitive tests under conjunction and negation.
That is, $T$ is the smallest set such that $T_0\subseteq T$, and if
$t_1$ and $t_2$ are in $T$, so is $t_1\land t_2$ and $\neg t_1$.
Thus, the language of tests is just a propositional language whose atomic
propositions are the elements of $T_0$. 

We consider two languages for choices.  In both cases, we begin with a
finite set $\A_0$ of \emph{primitive choices}.  These may be objects
such as `buy 100 shares of IBM' or `buy \$10,000 worth of bonds'. The
interpretation of these acts is tightly bound to the decision problem
being modeled.  The first language simply closes off $\A_0$ under
\ifthnels.  By this we mean that if $t$ is a test in $T$ and $a$ and $b$
are choices in $\A$, then $\ift{t}{a}{b}$ is also a choice in $\A$.
When we need to be clear about which $T_0$ and $A_0$ were used to
construct $\A$, we will write $\A_{\A_0,T_0}$.  Note that $\A$ allows
nesting, so that $\ift{t_1}{a}({\ift{t_2}{b}{c}})$ is also a
choice.

The second languages closes off $\A_0$ with \ifthnels\ and
randomization.  That is, we assume that objective probabilities are
available, and require that for any $0\leq r\leq 1$, if $a$ and $b$ are choices, so is $ra+(1-r)b$.  Randomization and \ifthnels\ can be
nested in arbitrary fashion.  We call this language $\A^+$
($\A^+_{\A_0,T_0}$ when necessary).

Tests in $T$ are elements of discourse about the world.  They could be events upon which choice is contingent:  If the noon price of  stock today is below \$600, then buy 100 shares, else buy none.  More generally, tests in $T$ are part of the DM's description of the decision problem, just as states are part of the description of the decision problem in Savage's framework.  However, elements of $T$ need not be complete descriptions of the relevant world, and therefore may not correspond to Savage's states. When we construct state spaces, elements of $T$ will clearly play a role in defining states, but, for some of our representation theorems, states cannot be constructed out of elements of $T$ alone.  Additional information in states is needed for both incompleteness of preferences and when the outcome space is taken to be objective or exogenously given.

The choices in $\A$ and $\A^+$ are \emph{syntactic} objects; strings of symbols.   They can be given \emph{semantics}---that is, they can be interpreted---in a number of ways.  For most of this paper we focus on one particular way of interpreting them that lets us connect them to Savage acts, but we believe that other semantic approaches will also prove useful (see Section~\ref{sec:conclusion}). The first step in viewing choices as Savage acts is to construct a state space $S$, and to interpret the tests as events (subsets of $S$).  With this semantics for tests, we can then construct, for the state space $S$ and a given outcome space $O$, a function $\rho_{SO}$ that associates with each choice $a$ a Savage act $\rho_{SO}(a)$, that is, a function from $S$ to $O$.  Given a state space $S$, these constructions work as follows:
\begin{definition} A \emph{test interpretation} $\pi_S$ for the state space $S$ is a function associating with each test a subset of $S$.  An interpretation is \emph{standard} if it interprets $\neg$ and $\land$ in the usual way; that is
\begin{itemize}
\item $\pi_S(t_1 \land t_2) = \pi_S(t_1) \intr \pi_S(t_2)$
\item $\pi_S(\neg t) = S - \pi_S(t)$.
\end{itemize}
\end{definition}

\noindent Intuitively, $\pi_S(t)$ is the set of states where $t$ is
true.  Up to Section~\ref{sec:nonstandard}, we assume that all
interpretations are \emph{standard}, that is, the obey the rules of
classical logic.  In particular, this means that in all states,
exactly one of $t$ or $\neg t$ is true.  
A standard interpretation is completely determined by its behavior on
primitive tests.

\begin{definition}A \emph{choice interpretation}
$\rhoso$ for the state space $S$ and outcome space $O$
assigns to each choice $a\in \A$ a (Savage) act, that is, a function
$\rhoso(a):S\to O$.
\end{definition}
\noindent Given a test interpretation $\pi_S$ and a choice
interpretation $\rhoso^0: \A_0 \to O^S$ for primitive choices,
which assigns to each $a_o\in \A_0$ a function from $S\to O$,
we can construct a choice interpretation by extending $\rhoso^0$
inductively as follows:
\begin{equation}\label{ifthenelse}
\rho_{SO}(\mbox{{\bf if} $t$ {\bf then} $a_1$ {\bf else} $a_2$})(s) =
\left\{\begin{array}{ll}
\rho_{SO}(a_1)(s) &\mbox{if $s \in \pi_S(t)$}\\
\rho_{SO}(a_2)(s) &\mbox{if $s \notin \pi_S(t)$.}
\end{array}\right.
\end{equation}
A choice interpretation $\rhoso$ constructed in this way is said to be
\emph{compatible with 
  $\pi_S$ (and $\rhoso^0$)}.
This semantics captures the idea of contingent choices; that, in the
choice \textbf{if $t$ then $a_1$ else $a_2$}, the realization of $a_1$
is contingent upon $t$, while $a_2$ is contingent upon `not $t$'.  Of
course, $a_1$ and $a_2$ could themselves be non-primitive programs,
with nested $\ifthnels$ statements.

Extending the semantics to the language $\A^+$, given $S$, $O$, and
$\pi_S$, requires us to associate with each choice $a$ an
Anscombe-Aumann (AA) act \cite{AA63}, that is, a function from
$S$
to probability measures on $O$.  Let $\Delta(O)$ denote the set of
probability measures on $O$ and let $\Delta^*(O)$ be the subset of
$\Delta(O)$ consisting of the probability measures that put probability one on an outcome.  Let $\rho_{SO}^0: \A_0 \rightarrow \Delta^*(O)^S$ be a choice interpretation for primitive choices that assigns to each $a_o\in \A_0$ a function from $S\to \Delta^*(O)$. Now we can extend $\rho_{SO}^0$ by induction on structure to all of $\A^+$ in the obvious way. For \ifthnels\ choices we use \eqref{ifthenelse};
to deal with randomization, define
\begin{equation*}
\rho_{SO}(r a_1 + (1-r) a_2)(s) = r \rho_{SO}(a_1)(s) +
(1-r)\rho_{SO}(a_2)(s).
\end{equation*}
That is, the distribution $\rho_{SO}(r a_1 + (1-r) a_2)(s)$ is the
obvious mixture of the distributions $\rho_{SO}(a_1)(s)$ and $\rho_{SO}(a_2)(s)$. Note that we require $\rho_{SO}$ to associate with each primitive choice in each state a single outcome (technically, a distribution that assigns probability 1 to a single outcome), rather than an arbitrary distribution over outcomes.  So primitive choices are interpreted as Savage acts, and more general choices, which are
formed by taking objective mixtures of choices, are interpreted as
AA acts.  This choice is largely a matter of taste.  We would get similar representation theorems even if we allowed $\rho_{SO}^0$ to be an arbitrary function from $\A$ to $\Delta(O)^S$.  However, this choice
does matter for our interpretation of the results; see Example~\ref{xam3} for further discussion of this issue.

\subsection{The DM's Theory of the World}\label{sec:theory}
The DM will typically have some knowledge about relationships between
various tests.  For example, a DM that can do propositional reasoning
will realize that $t_1 \land t_2$ is equivalent to $t_2 \land t_1$.  A 
DM may also have domain-dependent knowledge.   For example, if the DM
knows that interest rates will remain constant between periods 2 and 3,
and interest rates are either 4\% or 5\%, if $R_i(j)$ says that the
interest rate in period $i$ is $j$\%, then the DM knows that
$$(R_2(4) \lor R_2(5)) \land (R_2(4) \dimp R_3(4)) \land (R_2(5) \dimp
R_3(5)).$$

Formally, we add to the description of a decision problem a
\emph{theory}, that is, a set $\T \subseteq T$ of tests. 
\begin{definition} A test interpretation $\pi_S$
for the state space $S$
\emph{respects} a theory $\T$ iff  for all $t\in\T$,
$\pi_S(t)=S$.\end{definition} 
\noindent A theory represents the DM's view of the world.
Different people may,
however, disagree about what they take to be obviously true of the
world. Many people will assume that the sun will rise tomorrow.
Others, like Laplace, will consider the possibility that it will not. 

Choices $a$ and $b$ are \emph{equivalent with respect to a set
$\Pi$ of test interpretations} if, no matter what interpretation $\pi
\in \Pi$ is used, they are interpreted as the same function. For
example, in any standard interpretation, $\ift{t}{a}{b}$ is equivalent
to $\ift{\neg\neg t}{a}{b}$; no matter what the test $t$ and
choices $a$ and $b$ are, these two choices have the same input-output
semantics.

\begin{definition}For a set $\Pi$ of test interpretations, choices
$a$ and $b$ are $\Pi$-equivalent, denoted $a\equiv_\Pi b$, if for all
test interpretations $\pi \in \Pi$, if $\pi$ is a test interpretation on
a state space $S$, then for all outcome spaces $O$ 
  and choice interpretations $\rho_{SO}$ compatible with $\pi_S$, we have
$\rho_{SO}(a)=\rho_{SO}(b)$.\end{definition} 

\noindent
Denote by $\Pi_{\T}$ the set of all standard interpretations that
respect theory $\T$. Then\\ $\Pi_{\T}$-equivalent $a$ programs and $b$ are said to
be $\T$-equivalent, and we write $a\equiv_{\T} b$. 
Note that equivalence is defined relative to a given set $\Pi$ of
interpretations.  Two choices may be equivalent with respect to the set
of all standard interpretations that hold a particular test $t$ to be
true, but not equivalent to the larger set of all standard test
interpretations.

In this section, we demonstrate the power of our approach by developing 
some well-known examples of framing.  Framing problems appear when 
a DM solves inconsistently two decision problems that are designed by 
the modeler to be equivalent or that are obviously similar after recognizing 
an equivalence.  The fact that choices are syntactic objects allows us to 
capture framing effects.  This is an explicit virtue of providing a framework 
for decision theory that is closer to the natural language DMs might use in 
considering their choices. 

\xam\label{xam:framing} 
Consider the following well-known example of the effects of framing, 
due to McNeil et al.~(\citeyear{MPST82}). DMs are asked to choose 
between surgery or radiation therapy as a treatment for lung cancer.  
The problem is framed in two ways.  In the  \emph{survival frame}, 
DMs are told that, of 100 people having surgery, 90 live through the 
post-operative period, 68 are alive at the end of the first year, and 34 
are alive at the end of five years; and of 100 people have radiation 
therapy, all live through the treatment, 77 are alive at the end of the 
first year, and 22 are alive at the end of five years.  In the 
\emph{mortality frame}, DMs are told that of 100 people having 
surgery, 10 die during the post-operative period, 32 die by the end of 
the first year, and 66 die by the end of five years; and of 100 people 
having radiation therapy, none die during the treatment, 23 die by the 
end of the first year, and 78 die by the end of five years. Inspection 
shows that the outcomes are equivalent in the two frames---90 of 100 
people living is the same as 10 out of 100 dying, and so on.  Although 
one might have expected the two groups to respond to the data in 
similar fashion, this was not the case. While only 18\% of DMs prefer 
radiation therapy in the survival frame, the number goes up to 44\% 
in the mortality frame. 

We can represent this example in our framework as follows.  We assume
that we have the following tests:
\begin{itemize}
\item $RT$, which intuitively represents `100 people have radiation
therapy';
\item $S$, which intuitively represents `100 people have surgery';
\item $L_i(k)$, for $i = 0, 1, 5$ and $k = 0, \ldots, 100$, which
intuitively represents that $k$ out of 100 people live through the
post-operative period (if $i =0$), are alive after the first year (if
$i=1)$, and are alive after five years (if $i=5$);
\item $D_i(k)$, for $i = 0, 1, 5$ and $k = 0, \ldots, 100$, which is like $L_i(k)$, except `live/alive' are replaced by `die/dead'.
\end{itemize}
In addition, we assume that we have primitive programs $a_S$ and $a_R$ that represent `perform surgery' and `perform radiation theory'.  With these tests, we can characterize the description of the survival frame by the following
test $t_1$:
\begin{equation*}
(S \rimp L_0(90) \land L_1(68) \land L_5(34)) \land
(RT \rimp L_0(100) \land L_1(77) \land L_5(22)),
\end{equation*}
(where, as usual, $t \rimp t'$ is an abbreviation for $\neg(t \land
\neg t')$); similarly, the mortality frame is characterized by the following test $t_2$:
\begin{equation*}
(S \rimp D_0(10) \land D_1(32) \land D_5(66)) \land
(RT \rimp D_0(0) \land D_1(23) \land D_5(78)).
\end{equation*}

The choices offered in the McNeil et al.~experiment can be viewed as
conditional choices: what would a DM do conditional on $t_1$ (resp.,
$t_2$) being true.  Using ideas from Savage, we can capture the survival frame as a decision problem with the following two choices:
\begin{equation*}
\begin{array}{ll}
\mbox{{\bf if} $t_1$ {\bf then} $a_S$  {\bf else}  $a$,  and}\\
\mbox{{\bf if} $t_1$ {\bf then} $a_R$  {\bf else}  $a$,}
\end{array}
\end{equation*}
where $a$ is an arbitrary choice.  Intuitively, comparing
these choices forces the DM to consider his preferences between $a_S$
and $a_R$ conditional on the test, since the outcome in these two
choices is the same if the test does not hold.
Similarly, the mortality frame amounts to a decision problem with the
analogous choices with $t_1$ replaced by $t_2$.

There is nothing in our framework that forces a DM to identify the
tests $t_1$ and $t_2$; the tests $L_i(k)$ and $D_i(100-k)$ a priori
are completely independent, even if the problem statement suggests
that they should be equivalent. Hence there is no reason for a DM to
identify the choices $\ift{t_1}{a_S}{a}$ and $\ift{t_2}{a_S}{a}$.
As a consequence, as we shall see, it is perfectly consistent with our
axioms that a DM has the preferences
$\ift{t_1}{a_S}{a}\succ\ift{t_1}{a_R}{a}$ and $\ift{t_2}{a_R}{a}\succ\ift{t_2}{a_S}{a}$.
a sophisticated DM might
understand that $L_i(k)\dimp D_i(100-k)$, for $i=0,1,5$ and
$k=1,\ldots,100$.
\exam

We view it as a feature of our framework that it can capture this
framing example for what we view as the right reason: the fact that
DMs do not necessarily identify $L_i(k)$ and $D_i(100-k)$.
Of course, a sophisticated DM might
understand that $L_i(k)\dimp D_i(100-k)$, for $i=0,1,5$ and
$k=1,\ldots,100$.
Such a DM would include these tests in her theory $\T$.  If these tests are
in her theory, then she will make the same decision in both frames.
More precisely, we have 
$\ift{t_1}{a_S}{a} \equiv_\T \ift{t_2}{a_S}{a}$ and
$\ift{t_1}{a_R}{a} \equiv_\T \ift{t_2}{a_R}{a}$.
 The results of the McNeil et al.~experiment can be interpreted in our
 language as a failure by some DMs to have a theory that makes tests
 stated in terms of mortality data or survival data semantically equivalent.
A similar approach can be used to capture `nudges' or `naive
 diversification'. 

\xam\label{xam:nudge}  As discussed at length by Thaler and Sunstein
(\citeyear{TS09}), what the default choice is can have a significant
effect on the outcome of a decision problem.  For example, whether the
default in a company's 401K plan is that employees are enrolled in the plan, with the
possibility of opting out, or the default is that employees are not
enrolled in the plan, but have the option of opting in, can make a
huge difference to participation rates.

We model this in our framework using an approach similar to that
used in Example~\ref{xam:framing}.  There are two frames.  In the
first, the default is participation and the active option is to opt out; 
in the second, the default is not participating, and the active option is
opting in.  From the agent's perspective, there are two basic actions: 
\begin{itemize}
\item $a_1$: do nothing (which means that the agent will get the
  default option);
\item $a_2$: choose the active option.
\end{itemize}
Let $t_1$ be the formula that describes the first frame: if the agent
does nothing then he participates in the 401K; otherwise, he fills in
paperwork and opts out.  Let $t_2$ be the corresponding
formula that describes the second frame.  This means that the 
agent is comparing $\ift{t_1}{a_1}{a}$ and $\ift{t_1}{a_2}{a}$ in the
first frame, where $a$ is some default action, and comparing 
$\ift{t_2}{a_1}{a}$ and $\ift{t_2}{a_2}{a}$.  As it happens,
$\ift{t_1}{a_1}{a}$  and $\ift{t_2}{a_2}{a}$ (resp.,
$\ift{t_1}{a_2}{a}$ and $\ift{t_2}{a_1}{a}$) lead to the same outcome
in the actual domain.  A DM who is resource-bounded, however, may
not recognize this cross-frame equivalence and chooses $a_1$ in both
frames.
\exam

\xam Benartzi and Thaler (\citeyear{BT01}) observed that if
employees are offered $n$ choices in a defined contribution savings plan,
many naively diversify, using the heuristic of putting $1/n$ of
their allocation into each of the $n$ options offered.  This, to us,
shows the effect of language on choices. 

The survey that Benartzi and Thaler use to illustrate this point asks employees to allocate their retirement savings between two funds: A and B. In one instance, fund A is a stock fund and fund B is a bond fund. In the second instance, fund A is a stock fund and fund B is a balanced fund putting one-half of its money in stocks and one-half in bonds. Benartzi and Thaler find that a large portion of participants split their retirement savings equally between A and B in both instances. To test whether this results from a misunderstanding of stocks and bonds, Benartzi and Thaler run the same survey with the terms stock and bond replaced by the distribution on returns that would result from each choice. They get the same answer.

In our framework a simplified version of this survey can be captured by a decision problem with two primitive actions:
\begin{itemize}
\item $a_1$: put $1/2$ of the money in fund A and $1/2$ in fund B;
\item $a_2$: put all of the money in fund B;
\end{itemize}
and two tests that describe the contents of the funds:
\begin{itemize}
	\item $t_1$: the formula describing the first frame in which
          fund A is all stock and fund B is all bonds; 
	\item $t_2$: the formula describing the second frame in which fund A is all stock and fund B is one-half stock and one-half bonds.
\end{itemize}
The decision maker compares $\ift{t_1}{a_1}{a}$ and
$\ift{t_1}{a_2}{a}$ where $a$ is some default action, and compares
$\ift{t_2}{a_1}{a}$ and $\ift{t_2}{a_2}{a}$. In this setting
$\ift{t_1}{a_1}{a}$ and $\ift{t_2}{a_2}{a}$ lead to the same outcome.
Not having thought about it, a DM may not realize
$\ift{t_1}{a_1}{a}$ and $\ift{t_2}{a_2}{a}$ lead to the same outcome;
indeed, may even believe that
$\ift{t_1}{a_1}{a}$ and $\ift{t_2}{a_1}{a}$ lead to the same outcome.

Our approach restricts a theory to being a set of tests (although we
consider a generalization of this in Section~\ref{sec:nonstandard}.
The programs $\ift{t_1}{a_1}{a}$ and $\ift{t_2}{a_2}{a}$ may be
equivalent in the mind of the experimenter, and perhaps also 
in the mind of the subject.  Our language of tests is not rich enough to 
capture how the DM interprets programs, but this is captured by the
DM's preference relation.  If the DM views two programs as equivalent,
then he will be indifferent between them.
%
\exam

\xam
In a famous experiment,  Johnson \emph{et al.}
(\citeyear{jhmk}) have a dramatic finding, tied to a
decision problem: Subjects offered hypothetical health insurance were
willing to pay a higher premium for policies covering hospitalization
for any disease than they were for policies covering
hospitalization for any reason at all, which would include both disease 
and accidents. 
\citet{tverskyetal94} understand this as a failure of
\emph{extensionality}.  They observe (p.~548) that `\ldots
probability judgments are attached not to events but to descriptions
of events'.  Extensionality is the property that different
descriptions of the same event should be assigned the same probability. 
It fails in the Johnson \emph{et 
al.} experiment because the event `any disease or accident' is a
strict subset of `any reason at all'; extensionality would require the
probability of `any reason at all' to be the sum of `any disease or
accident' and `a reason other than disease of accident'.
Tversky and Koehler respond
to this failure by introducing \emph{support theory} as a way of
putting likelihoods not on events but on their descriptions.    

In our framework, extensionality fails when two 
descriptions of the same event are not perceived as such by the DM.
Because for us states are subjective, representing the \emph{DM's}
point of view, so even if tests $t_1$ and $t_2$ are equivalent from
the modeler's point of view, we may have 
$\pi_S(t_1)\neq\pi_S(t_2)$ for the DM's state space $S$ and
test interpretation $\pi_S$. Thus, the DM may assign $t_1$ and $t_2$
different probabilities.

The DM's theory captures equivalence among descriptions from the point
of view of the DM.  For instance, suppose a DM is considering the purchase
of an home insurance policy.  One alternative might offer coverage up to the
value of the home and its contents `in case of damage'.  Another might 
offer the same coverage for `storm damage', `fire damage', and
`other damage'.  The relevant tests are $t_1$, damage of any kind, 
and $t_2$ through $t_4$: storm, fire, and other damage, 
respectively.  If $t_1 \dimp (t_2\lor t_3\lor t_4)$ is not part of her theory $\T$, then the DM's decision problem might
be represented by a state space $S$ and a test interpretation $\pi_S$
such that $\pi_S(t_1)\neq\pi_S(t_2)\cup\pi_S(t_3)\cup\pi_S(t_4)$.  On the other
hand, should
$t_1\dimp (t_2\lor t_3\lor t_3)$ be part of her theory $\T$, then for any
state space $S$ and test interpretation $\pi_S\in \Pi_{\T}$,
$\pi_S(t_1)=\pi_S(t_2)\cup\pi_S(t_3)\cup\pi_S(t_4)$. 

Support theory develops this idea in a very specific way.  Two
events $A$ and $B$ are \emph{exclusive} (with respect to an
interpretation $\pi_S$) if, in our terminology,  
$\pi_S(A)\cap\pi_S(B)=\emptyset$.  Support theory distinguishes two
kinds of disjunctions.  If $t_1$ and $t_2$ are exclusive, their
\emph{explicit disjunction} is $t_1\lor t_2$.   So in the preceding
paragraph, `$t_2$ or $t_3$ or $t_4$' , that is, `$t_2\lor t_3\lor
t_4$', is an explicit disjunction.  On
the other hand, `$t_1$' is an \emph{implicit disjunction}; it is in
fact the disjunction of $t_2$ and $t_3$ but
it is not an `or' statement.  Support
theory assumes the 
existence of a ratio scale $s$ that measures the degree of support for
an event.  The key idea of support theory is expressed in Tversky
and Koehler's (\citeyear{tverskyetal94}) equation 2: 
\begin{equation}\label{eq:TK}
s(t_1)\leq s(t_2\lor t_3\lor t_4)=s(t_2)+s(t_3)+s(t_4),
\end{equation}
that is, the scale $s$ is additive over explicit disjunction, but
subadditive over implicit disjunction.  
Effectively, Tversky and Koehler are assuming that if $t_1$ is an
implicit disjunction of $t_2$ and $t_3$, then the agent realizes that
$t_2 \rimp t_1$ and that $t_3 \rimp t_1$, but not that $t_1 \dimp (t_2
\lor t_3)$; that is, the first two formulas are in his theory, while
the third is not.  Under these assumptions, (\ref{eq:TK}) seems
perfectly reasonable.

`Implicit subadditivity' can be captured in our framework simply by allowing
$\pi_S(t_1) \supseteq
\pi_S(t_2)\cup \pi_S(t_3) \cup \pi_S(t_4)$.  Subsequent iterations of support
theory propose that subadditivity may also apply to explicit disjunctions,
that is, $s(t_2\lor t_3\lor t_4)\leq s(t_2)+s(t_3)+s(t_4)$.  In our approach,
this stronger form of 
subadditivity (and, more generally, non-additivity) can be captured by
considering 
nonstandard test interpretations; see Section~\ref{sec:nonstandard}.

We are not taking a stand on the normative implications of extensionality 
failures.  Our point is that extensionality and its
failure is a product of the theory of the world $\T$ that a DM brings to
the decision problem at hand (and how it compares to the
experimenter's or modeler's view of the world).  One can extend
conventional expected utility to state spaces and act spaces that
admit failures of extensionality.  A failure of extensionality is not
inconsistent with the existence of an SEU representation on a suitably
constructed state space.  Moreover, if we modelers take $t_0\dimp t_1$
to be an axiom that describes the world, but our DM disagrees, then
the probability of the set
$(\pi_S(t_0)/\pi_S(t_1))\cup(\pi_S(t_1)/\pi_S(t_0))$ measures the
degree of framing bias from the modeler's point of view.  Furthermore,
our framework offers a natural way of bringing support theory to
decision problems; a DM's degree of support for a formula could be
inferred from individuals' willingness to accept certain bets. 

In many experiments (e.g., that of Johnson \emph{et al.} (\citeyear{jhmk})),
some subjects are exposed to a treatment where they see only 
implicit disjunctions, while others are
exposed only to the basic events.  Different subjects
are exposed to different frames, and different frames come with
different languages, so there is no cause for surprise that
extensionality fails; options available in one treatment
are not present in the other  treatment.  It
is also interesting to ask what happens when language changes during
the course of evaluation.  We might imagine a sequential problem in
which subjects are first asked to make a broad assessment of why a
car failed, and then are asked for more explicit assessments about the
source of the failure.  The second-stage questions add possibilities to
the DM's language.  We could certainly model this in a probabilistic
way, which would show no subadditivity, but we can also imagine
alternatives that would be consistent with failures of extensionality.
This, however, is a topic for another paper. 
\exam

\section{The Axioms}\label{sec:representation}
This section lays out our basic assumptions on preferences.  If we
want to get an SEU representation, then we must have an analogue of an
independence axiom or the sure-thing principle. However, in their
usual form, these axioms require a mixture space.  The set of feasible
programs is not a mixture space.  A standard way  of dealing with this
problem (see, e.g., \cite{KLST71}) is to use \emph{cancellation
axioms}.  Since cancellation is not so well known among
economists, we also illustrate the relationship between cancellation and more
familiar preference properties.

\subsection{Preferences}
We assume that the DM has a preference relation $\succeq$ on a
subset $C$ of the set $\A$ (resp., $\A^+$) of non-randomized (resp.,
randomized) acts.\footnote{We do not assume that this relation is necessarily complete or transitive; it is just a binary relation on $C\times C$.} This preference relation has the usual
interpretation of `at least as good as'.  We take $a \succ b$ to be an abbreviation for $a\succeq b$ and $b {\not{\succeq}} \, a$, even if $\succeq$ is not complete.  We prove various representation theorems that depend upon the language, and upon whether outcomes are taken to be given or not.  The engines of our analysis are various cancellation axioms, which are the subject of the next section.  At some points in our analysis we consider complete preferences:
\begin{axiom}\label{Aone}
The preference relation $\succeq$ is complete.
\end{axiom}
\noindent The completeness axiom has often been defended by the claim
that `people, in the end, make choices'.  Nonetheless, from the outset
of modern decision theory, completeness has been regarded as a
problem. \citet[Section 2.6]{Savage} discusses the difficulties
involved in distinguishing between indifference and incompleteness.  He
concludes by choosing to work with the relationship he describes with
the symbol $\leqq\kern-1pt\raisebox{2pt}{$\mathbf{\cdot}$}$, later
abbreviated as $\leq$, which he interprets as `is not preferred to'.
The justification of completeness for the `is not preferred to'
relationship is anti-symmetry of strict preference. Savage,
\citet{Au62}, \citet{Bewley} and \citet{mandler} argue
against completeness as a requirement of rationality.
\citet{eliazok} have argued that rational choice theory with
incomplete preferences is consistent with preference reversals.  In our
view, incompleteness is an important expression of ambiguity in its
plain meaning (rather than as a synonym for a non-additive
representation of likelihood).  There are many reasons why a comparison
between two objects of choice may fail to be resolved:  obscurity or
indistinctness of their properties, lack of time for or excessive cost
of computation, the incomplete enumeration of a choice set, and
so forth.  We recognize indecisiveness in ourselves and others, so it
would seem strange not to allow for it in any theory of preferences that
purports to describe tastes (as opposed to a theory which purports to
characterize consistent choice).

We will focus our analysis on preferences that may be incomplete. To
us, requiring completeness for preferences over the complex syntactic objects
in $\A$ seems unnaturally restrictive. Our axioms thus yield representation 
theorems for incomplete preferences. Once we obtain states, outcomes, and
acts, our theorems are similar to other representation theorems for
incomplete preferences (e.g., \cite{DMO04,GK2013,Nau2006}),
although our axioms differ as they are placed on
preferences over different objects.
While representation theorems with incomplete preferences are, by
now, quite standard, we are able to use them to give some interesting
insight into updating (see Section~\ref{sec:updating}).

\subsection{Cancellation} 
Although simple
versions of the cancellation axiom have appeared in the literature
(e.g., \citet{scott64} and \citet{KLST71}), it is nonetheless
not well known, and so before turning to our framework we briefly
explore some of its implications in more familiar settings.
Nonetheless, some of the results here are new; in particular, the results on
cancellation for partial orders.  These will be needed for proofs in the
appendix.

Let $C$ denote a set of choices and $\succeq$ a preference relation on~$C$.  We use the following notation:  Suppose
$\< a_1,\ldots,a_n\>$ and $\langle b_1,\ldots,b_n\rangle$ are
sequences of elements of $C$.  If for all $c\in C$, $\#\{j:a_j=c\}=
\#\{j:b_j=c\}$, we write $\{\{a_1,\ldots a_n\}\}=
\{\{b_1,\ldots,b_n\}\}$.  That is, the \emph{multisets} formed by the two sequences are identical.

\begin{definition}[Cancellation] The preference relation $\succeq$ on
$C$ satisfies \emph{cancellation} iff for all pairs of sequences
$\<a_1,\ldots,a_n\>$ and $\<b_1,\ldots,b_n\>$ of elements of $C$ such that $\{\{a_1,\ldots,a_n\}\} = \{\{b_1,\ldots,b_n\}\}$, if $a_i\succeq b_i$ for $i\leq n-1$, then $b_n\succeq a_n$.
\end{definition}
\noindent Roughly speaking, cancellation says that if two collections of choices are identical, then it is impossible to order the choices so as to prefer each choice in the first collection to the corresponding choice in the second collection.  The following proposition shows that cancellation is equivalent to reflexivity and transitivity.
Although \citet[p.~251]{KLST71} and
\citet[p.~743]{Fishburn92} have observed that cancellation implies
transitivity, this full characterization appears to be new.
\begin{prop}\label{prop:cancel} A preference relation $\succeq$ on a
choice set $C$ satisfies cancellation iff
\begin{thmlist}
\item[(a)] ${}\succeq{}$ is reflexive, and
\item[(b)] $\succeq$ is transitive.
\end{thmlist}
\end{prop}


\noindent {\bf All proofs are provided in the Appendix.}

\commentout{
Our representation theorems will all use some version of a cancellation axiom. As we provide representation theorems for decision makers who may be subject to framing effects, who may have theories of the world that seem obviously wrong to an outside observer or even who are not logically omniscient (Section 5), it may seem incongruous to require them to satisfy cancellation. However, it's worth noting that in its simplest form cancellation is equivalent to reflexivity and transitivity which are obviously necessary for a representation theorem. For our representation theorems on mixture spaces we need an extended cancellation axiom but it is equivalent to reflexivity, transitivity and rational mixture independence. A decision maker who does not satisfy cancellation simply has preferences that do not have a SEU representation. One advantage of our framework is that this observation (the necessity part of our representation theorems in Section 4) is not conditional on a particular state space, outcome space and act interpretation of the problem. Instead, it is inherent in the decision maker's preferences over the objects of choice. 
}

We use two strengthenings of cancellation in our representation
theorems for $\mathcal{A}$ and $\mathcal{A^+}$, respectively.  The
first, statewise cancellation, simply applies cancellation to a
setting where the objects of choice are functions; in that case, we
apply cancellation to each argument of the function.  
We first state the statewise cancellation to Savage acts, which we
view as functions from a finite set $S$ of states to a finite set $O$
of outcomes.
Let $C$ denote a set of Savage acts and suppose that $\succeq$ is a
preference relation on~$C$. 

\begin{definition}[Statewise Cancellation] The preference relation
$\succeq$ on a set $C$ of Savage acts satisfies \emph{statewise
cancellation} iff for all pairs of sequences $\<a_1,\ldots,a_n\>$ and $\<b_1,\ldots,b_n\>$ of elements of $C$, if $\{\{a_1(s), \ldots, a_n(s)\}\}=\{\{b_1(s), \ldots, b_n(s)\}\}$ for all $s \in S$, and $a_i \succeq b_i$ for $i\leq n-1$, then $b_n \succeq a_n$.
\end{definition}

As written, statewise cancellation is really an infinite family of
postulates, one for each $n$.  As we now show, we actually only need
cancellation to hold for finitely many values of $n$.  Let $SC_n$ be
the instance of cancellation that involves pairs of sequences of
length $n$.

\begin{prop}\label{prop:cancellationbound} If 
  $\succeq$ is a preference relation on a finite set $C$, then there
  exists $N$ such that $\succeq$ satisfies $SC_n$, $n = 
  1, 2, 3, \ldots$ iff $\succeq$ satisfies
  $SC_n$ for all $n \le N$.
\end{prop}

\noindent The argument used to prove
Proposition~\ref{prop:cancellationbound} applies with essentially no
change to all later variants of cancellation that we consider.  We believe
it should be possible to obtain a bound on the number of instances of
the cancellation postulate needed in terms of the cardinality of the
set $C$ of Savage acts, although we have not proved this.

\Csa\ is a powerful assumption because
equality of the multisets is required only `pointwise'.  Any pair of
sequences that satisfy the conditions of  \ca\ also satisfies
the conditions of \csa, but the converse is not true.  For instance, suppose that $S = \{s_1,s_2\}$, and we use $(o_1,o_2)$ to refer to an act with outcome $o_i$ in state $i$, $i=1,2$.  Consider the two sequences of acts $\langle (o_1,o_1), (o_2,o_2)\rangle$
and $\langle (o_1,o_2), (o_2,o_1)\rangle$.  These two sequences satisfy the conditions of \csa, but not that of \ca.

In addition to the conditions in Proposition~\ref{prop:cancel},
statewise cancellation directly implies \emph{event independence}, a
condition at the heart of SEU representation theorems (and which can be used to derive the Sure Thing Principle). If $T \subseteq S$, let $a_Tb$ be the Savage act that agrees with $a$ on $T$ and with $b$ on $S-T$; that is $a_T b(s) = a(s)$ if $s \in T$ and $a_T b(s) = b(s) $ if $s \notin T$. We say that $\succeq$ satisfies \emph{event independence} iff for all acts $a$, $b$, $c$, and $c'$ and subsets $T$ of the state space $S$, if $a_T c \succeq b_T c$, then $a_T c' \succeq b_T c'$.

\begin{prop}\label{prop:savagecancel}
If $\succeq$ satisfies statewise cancellation, then $\succeq$ satisfies event independence.
\end{prop}

\noindent Proposition~\ref{prop:cancel} provides a
characterization of cancellation for choices in terms of familiar
properties of preferences. 
%
We do not have a
similarly
simple characterization of statewise cancellation. In particular,
the following example shows that
it is not equivalent to the combination of reflexivity and transitivity
of $\succeq$ and event independence.

\xam\label{xam:charcounterexample} Suppose that $S =
\{s_1,s_2\}$,
$O = \{o_1,o_2,o_3\}$.
There are nine possible acts.  Suppose that $\succeq$ is the smallest
reflexive, transitive relation such that
\begin{gather*}
(o_1,o_1) \succ (o_1,o_2) \succ (o_2,o_1) \succ (o_2,o_2) \succ (o_3,o_1)
\succ\\ (o_1,o_3) \succ (o_2,o_3) \succ (o_3,o_2) \succ (o_3,o_3),
\end{gather*}
using the representation of acts described above.  To see that $\succeq$ satisfies event independence, note that
\begin{itemize}
\item $(x,o_1) \succeq (x,o_2) \succeq (x,o_3)$ for $x\in\{o_1,o_2,o_3\}$;
\item $(o_1,y) \succeq (o_2,y) \succeq (o_3,y)$ for $y\in\{o_1,o_2,o_3\}$.
\end{itemize}
However, statewise cancellation does not hold.  Consider the sequences
\begin{equation*}
\<(o_1,o_2),(o_2,o_3),(o_3,o_1)\>\mbox{ and }
\<(o_2,o_1),(o_3,o_2),(o_1,o_3)\>.
\end{equation*}
This pair of sequences clearly satisfies the hypothesis of statewise cancellation, that $(o_1,o_2) \succeq (o_2,o_1)$ and $(o_2,o_3)\succeq (o_3,o_2)$, but $(o_1,o_3){\not{\succeq}} (o_3,o_1)$.
\exam

For our representation theorems for complete orders, statewise
cancellation suffices.  However, for partial orders, we need a version of cancellation that is equivalent to statewise cancellation in
the presence of \Aone, but is in general stronger.

\begin{definition}[Extended Statewise Cancellation] The preference relation $\succeq$ on a set $C$ of Savage acts satisfies \ecsa\
if and only if for all pairs of sequences $\<a_1,\ldots,a_n\>$ and
$\<b_1,\ldots,b_n\>$ of elements of $C$ such that\\ $\{\{a_1(s), \ldots, a_n(s)\}\}=\{\{b_1(s),\ldots,b_n(s)\}\}$ for all $s \in S$,
if there exists some $k<n$ such that $a_i \succeq b_i$
for $i\leq k$, $a_{k+1}=\cdots=a_n$, and $b_{k+1}=\cdots=
b_n$, then $b_n \succeq a_n$.
\end{definition}

With extended statewise cancellation, we require that the last $n-k$
elements in each of the sequences are equal.  The fact that $b_n
\succeq a_n$ thus means that $b_j \succeq a_j$ for $j = k+1, \ldots, n$.
Note that statewise cancellation is the special case of extended
statewise calculation where $k = n-1$.  

\begin{prop}\label{prop:cancelequivalent}
In the presence of \Aone, extended statewise cancellation and statewise cancellation are equivalent.
\end{prop}

The extension of cancellation needed for $\A^+$ is based on the same
idea as extended statewise cancellation, but probabilities of objects
rather than the incidences of objects are added up.  Let $C$ denote a collection of elements from a finite-dimensional mixture space.
Thus, $C$ can be viewed as a subspace of $\Nspace{n}$ for some $n$, and
each component of any $c\in C$ is a probability.  We can then formally
`add' elements of $C$, adding elements of $\Nspace{n}$ pointwise.  (Note
that the result of adding two elements in $C$ is no longer an element of
$C$, and in fact is not even a mixture.)

\begin{definition}[Extended Mixture Cancellation]
The preference relation $\succeq$ on $C$ satisfies \emph{extended
mixture cancellation} iff for all pairs of sequences $\<a_1,
\ldots,a_n\>$ and $\<b_1,\ldots, b_n\>$ of elements of $C$,
such that $\sum_{i=1}^n a_i = \sum_{i=1}^n b_i$,
if there exists some $k<n$ such that $a_i\succeq b_i$
for $i\leq k$, $a_{k+1}=\cdots=a_n$, and $b_{k+1}=\cdots=
b_n$, then $b_n \succeq a_n$.
\end{definition}
\noindent We can extend Proposition~\ref{prop:cancel} to get a characterization theorem for preferences on mixture spaces by using an independence postulate.  The preference relation $\succeq$ satisfies
\emph{mixture independence} if for all $a$, $b$, and $c$ in $C$, and all $r\in(0,1]$, $a \succeq b$ iff $ra+(1-r)c\succeq r b+(1-r)c$.
The preference relation $\succeq$ satisfies \emph{rational mixture
  independence} if it satisfies mixture independence for all rational
$r \in (0,1]$.

\thm\label{thm:cancellationchar1}
A preference relation $\succeq$ on a finite-dimensional mixture space
$C$ satisfies \ecm\ iff $\succeq$ is reflexive, transitive, and satisfies rational mixture independence.
\ethm

\subsection{The cancellation postulate for
  choices}\label{sec:cancallationchoices} 
We use cancellation to get a representation theorem for preference
relations on choices. However, the definition of the cancellation
postulates for Savage acts and mixtures
rely on (Savage) states. We now develop an analogue of this postulate
for
our syntactic notion of choice.
\begin{definition} Given a set $T_0 = \{t_1, \ldots, t_n\}$ of primitive tests, an {\em atom over $T_0$\/} is a test of the form $t_1'\land\ldots \land t_n'$, where $t_i'$ is either $t_i$ or $\neg t_i$.\end{definition}

An atom is a possible complete description of the truth value of tests according to the DM.  If there are $n$ primitive tests in $T_0$, there are $2^n$ atoms.  Let $\Atoms(T_0)$ denote the set of atoms over $T_0$. It is easy to see that, for all tests $t \in T$ and atoms $\delta \in \Atoms(T_0)$, and for all state spaces $S$ and standard test interpretations $\pi_S$, either $\pi_S(\delta)\subseteq \pi_S(t)$ or
$\pi_S(\delta)\cap\pi_S(t)=\emptyset$.  (The formal proof is by
induction on the structure of $t$.) We write $\delta \rimp t$ if the
former is the case. We remark for future reference that a standard test
interpretation is determined by its behavior on atoms. (It is, of
course, also determined completely by its behavior on primitive
tests).

\begin{definition}An atom $\delta$ (resp., test $t$) is
\emph{consistent with
    a theory $\T$} if there exists a state space $S$ and a test
interpretation $\pi_S\in\Pi_{\T}$ such that $\pi_S(\delta) \ne
\emptyset$ (resp., $\pi_S(t) \ne \emptyset$). Let $\AtomsAX(T_0)$
denote the set of atoms over $T_0$ consistent with $\T$. 
\end{definition}

\noindent Intuitively, an atom $\delta$ is consistent with $\T$ if there
is some state in some state space where $\delta$ might hold, and
similarly for a test $t$.
The problem of checking whether an atom or a test is consistent with a
theory $\T$ is NP-complete (it is basically the problem of testing
whether a formula is satisfiable), that is, it is hard for the DM to do.\footnote{Satisfiability may or may not be relevant for the DM, depending on the task at hand.}

A choice in $\A$ can be identified with a  function from atoms to
primitive choices in an obvious way.  For example, if $a_1$, $a_2$,
and $a_3$ are primitive choices and $T_0 = \{t_1, t_2\}$, then the
choice $a=\ift{t_1}{a_1}{(\ift{t_2}{a_2}{a_3})}$ can be identified
with the function $f_a$ such that
\begin{itemize}
\item $f_a(t_1 \land t_2) = f_a(t_1 \land \neg t_2) = a_1$;
\item $f_a(\neg t_1 \land t_2) = a_2$; and
\item $f_a(\neg t_1 \land \neg t_2) = a_3$.
\end{itemize}
Formally, we define $f_a$ by induction on the structure of choices.  If $a \in \A_0$, then $f_a$ is the constant function $a$, and
\begin{equation*}
f_{\mbox{\textbf{\scriptsize if $t$ then $a$ else $b$}}}(\delta)=
\begin{cases}
f_a(\delta) &\mbox{if $\delta \rimp t$}\\
f_b(\delta) &\mbox{otherwise.}
\end{cases}
\end{equation*}

We consider a family of cancellation postulates, relativized to the
axiom system $\T$. The cancellation postulate for $\T$
(given the language $\A_0$) is simply statewise cancellation for Savage acts, with atoms consistent with $\T$ playing the role of states.

\begin{atwo}If $\<a_1,\ldots, a_n\>$ and
$\<b_1, \ldots, b_n\>$ are two sequences of choices in $\A_{\A_0,T_0}$
such that for each atom $\delta \in \AtomsAX(T_0)$,
 $\{\{f_{a_1}(\delta), \ldots,
f_{a_n}(\delta)\}\}=\{\{f_{b_1}(\delta),\ldots, f_{b_n}(\delta)\}\}$,
and there exists some $k<n$ such that
$a_i \succeq b_i$  for all $i\leq k$, $a_{k+1}=\cdots=a_n$,
and $b_{k+1}=\cdots=b_n$, then $b_n \succeq a_n$.\end{atwo}
\noindent We drop the prime and refer to $\AthreeAX$ when $k=n-1$.

Axiom $\AthreeAX$ implies the simple cancellation of the last
section, and so the conclusions of Proposition~1 hold: $\succeq$ on
$\mathcal{A}$ will be transitive and reflexive.  $\AthreeAX$\ has another consequence: a DM must be indifferent between $\T$-equivalent choices.

\begin{prop}\label{prop:equivalence} Suppose that $\succeq$ satisfies
$\AthreeAX$.  Then $a \equiv_{\T} b$ implies $a \sim b$.
\end{prop}

Proposition \ref{prop:equivalence} implies that the behavior
of $a$ and $b$ on atoms not in $\AtomsAX(T_0)$ is irrelevant; that is,
they are null in Savage's sense.  We define this formally:

\begin{definition}A test $t$ is \emph{null} if, for all acts $a$, $b$
and $c$, $\ift{t}{a}{c}\sim\mbox{\textbf{\textit{ if $t$}}}$
\textbf{then $b$ else $c$}.
\end{definition}

\noindent An atom (or test) inconsistent with the theory $\T$ must be null 
(so we can test if the DM holds the theory $\T$ by checking whether, for all
$\neg t\in \T$, $t$ is null), but 
consistent tests may be null as well.
If we add as an axiom that no test consistent with the theory can be
null, then (at least in principle) a DM's theory is testable.  

The notion of a null test is suggestive of, more generally, test-contingent preferences.
\begin{definition} If $t$ is a test in $T$, then for any acts $a$ and
$b$, $a\succeq_tb$ iff for some $c$, $\ift{t}{a}{c}\succeq
\ift{t}{b}{c}$.\end{definition}

\noindent Proposition~\ref{prop:savagecancel} shows that statewise
cancellation implies that the choice of $c$ is irrelevant, and so
test-contingent preferences are well-defined.
Before giving our representation theorems, we compare our framework to
that of \citet{Lip99}.  Lipman starts with a collection of
preference relations $\preceq_I$ indexed by what he calls
\emph{information sets} $I$.  For Lipman, an information set is a
possible piece of information that a DM might receive.  An
information set can be identified with a test in our framework (except
that information sets do not have the syntactic structure of tests).
Lipman takes a possible world to be a maximal piece of information;
Lipman's possible worlds are essentially our atoms.  Lipman takes the
DM's subjective state space to be the set of possible worlds.  As we show
(Theorem~\ref{thm:eurepa}), if we assume $\Aone$ (as Lipman does),
then we can also take the DM's subjective state space in our
representation to be the set of atoms consistent with the DM's theory $\T$.
Otherwise, without $\Aone$, if we want a representation with a single
utility function, we cannot in general take the state space to be just
the set of atoms consistent with the DM's theory $\T$.

To get a representation theorem for $\A^+$, we use a mixture
cancellation postulate, again replacing states by atoms.  The idea now
is that we can identify each choice $a$ with a function $f_a$
mapping atoms consistent with $\T$ into distributions over primitive
choices.  For example, if $t$ is the only test in $T_0$ and $\T=
\emptyset$, then the choice
$a = \frac{1}{2} a_1 +  \frac{1}{2}(\ift{t}{a_2}{a_3})$
can be identified with the function $f_a$ such that
\begin{itemize}
\item $f_a(t)(a_1) = 1/2$; $f_a(t)(a_2) = 1/2$
\item $f_a(\neg t)(a_1) = 1/2$; $f_a(\neg t)(a_3) = 1/2$.
\end{itemize}
Formally, we just extend the definition of $f_a$ given in the previous
section by defining
\begin{equation*}
f_{ra_1 + (1-r)a_2}(\delta) = rf_{a_1}(\delta) + (1-r)f_{a_2}(\delta).
\end{equation*}

\noindent Consider the following cancellation postulate:

\noindent $\AthreeprpAX$. If $\<a_1,
\ldots, a_n\>$ and $\<b_1, \ldots, b_n\>$ are two sequences of acts in
$\A_{\A_0,T_0}^+$ such that
\begin{equation*}
f_{a_1}(\delta) + \cdots +  f_{a_n}(\delta) = f_{b_1}(\delta) + \cdots +
f_{b_n}(\delta)
\end{equation*}
for all atoms $\delta$ consistent with $\T$, and there exists $k < n$ such that $a_i\succeq b_i$ for $i\le k$, $a_{k+1}=\ldots=a_n$, and
$b_{k+1}=\ldots=b_n$, then $b_n \succeq a_n$.

\noindent Again, $\AthreeprpAX$ can be viewed as a generalization of
$\AthreeprAX$.\footnote{$\AthreeprpAX$ is analogous to extended
  \emph{statewise} mixture cancellation.  It may seem strange that we
  need the cancellation to be statewise, since extended statewise
  mixture cancellation is equivalent to extended mixture cancellation.
  This suggests that we might be able to use a simpler non-statewise
  axiom.  However, the obvious non-statewise version of the axiom and
    its statewise version are  equivalent only if we assume that all
  that matters about a choice is how it acts as a function from atoms
  to primitive choices:  If $f_a = f_b$, then $a \sim b$. Rather than
    add this axiom, we use $\AthreeprpAX$.}

We use $\AthreeprpAX$ or $\AthreeprAX$ in our representations
theorems.  Since we consider DMs who may have some
obvious logical blindspots (and in Section~\ref{sec:nonstandard} even
consider DMs who may not satisfy the basic axioms of propositional
logic), it may seem unreasonable to expect an agent to satisfy these
postulates.  We do not think it is always so unreasonable.  For example,
as we show in the proof of Theorem~\ref{thm:rep4}, if we assume
completeness, in the  Anscombe-Aumann setting, cancellation is
equivalent to a number of other standard assumptions in the
literature.  Moreover, the preferences of a DM who is maximizing
expected utility with respect to her subjective representation of the
world will satisfy cancellation, even
if the DM is not aware of the property and is not actively trying to
ensure that her preferences are consistent with it.  In any case, as we show,
cancellation (or properties equivalent to
it) is needed to get an SEU-like representation.  So agents that do
not satisfy it will simply not act as utility maximizers
(even with respect to their subjective representation of the world).

\section{Representation Theorems}
Having discussed our framework, we now turn to the representation
theorems.  Our goal is to be as constructive as possible.  In this
spirit we want to require that preferences exist not for all possible
acts that can be described in a given language, but only for those in a
given subset, henceforth designated~$C$.  We are agnostic about the
source of~$C$.  It could be the set of
choices in one particular decision problem, or it could be the set of
choices that form a universe for a
collection of
decision problems.
One cost of our finite framework is that we will have no uniqueness
results.  In our framework the preference representation can fail to be
unique because of our freedom to choose different state and outcome
spaces, but even given these choices, the lack of richness of $C$ may
allow multiple representations of the same (partial) order.

\subsection{A Representation Theorem for $\A$}
By a representation for a preference relation on $\A$ we mean the following:
\begin{definition}A preference relation on a set $C\subseteq
\A_{\A_0,T_0}$ has a \emph{constructive $\T$-consistent SEU representation} iff
there is a finite set of states $S$, a finite set $O$ of outcomes, a set $\U$ of utility functions $u:O\to\R$, a set $\mathcal{P}$ of probability distributions on $S$, a subset $\V\subseteq \U\times\P$,
a test interpretation $\pi_S$ consistent with $\T$, and a choice
interpretation $\rho_{SO}$ such that $a\succeq b$ iff
\begin{equation*}
\sum_{s\in S}u\bigl(\rho_{SO}(a)(s)\bigr)p(s)\geq
\sum_{s\in S}u\bigl(\rho_{SO}(b)(s)\bigr)p(s) \
\mbox{for all}\ (u,p) \in \V.
\end{equation*}
\end{definition}
We are about to claim that $\succeq$ satisfies $\AthreeprAX$ if and
only if it has a constructive $\T$-consistent representation.  In the
representation, we have a great deal of flexibility as to the choice
of the state space $S$ and the outcome space $O$.  One might have
thought that the space of atoms, $\AtomsAX(T_0)$, would be a rich
enough state space on which to build representations.  This is not
true when preferences are incomplete  
and we ask for a representation with a single utility function.  A
rich enough state space is needed to account for the incompleteness. 
\begin{definition} Given a partial order $\succeq$ on a set $C$ of
choices, let $\EXAX(\succeq)$ denote all the extensions of $\succeq$
to a total order on $C$ satisfying $\AthreeAX$.
\end{definition}
Our proof shows that we can take $S$ to be  $\AtomsAX(T_0) \times
\EXAX(\succeq)$.  Thus, in particular, if $\succeq$ is complete, then we
can take the state space to be $\AtomsAX(T_0)$.  We later give examples
that show that if $\succeq$ is not complete then, in general, 
if we ask for a representation with a single utility function, then the state
space must have cardinality larger than that of $\AtomsAX(T_0)$.
While for some applications there may be a more natural state space,
our choice of state space shows that we can always view the DM's
uncertainty as stemming from two sources: the truth values of various
tests (which are determined by the atoms) and the relative order of two
choices not determined by $\succeq$ (which is given by the extension
${}\succeq'\,\in\EXAX(\succeq)$ of $\succeq$).
The idea of a DM being uncertain about her preferences is prevalent
elsewhere in decision theory; for instance, in the menu choice literature
\cite{kreps79}.  This uncertainty can be motivated in any number of
ways, including both incomplete information and resource-bounded
reasoning.

\begin{theorem}\label{thm:eurepa} 
A preference relation $\succeq$ on a set $C\subseteq \A_{\A_0,T_0}$ has a constructive $\T$-consistent SEU representation iff $\succeq$ satisfies $\AthreeprAX$.  Moreover, in the representation, either $\P$ or $\U$ can be taken to be a singleton and, if $\U$ is a singleton $\{u\}$, the state space can be taken to be $\AtomsAX(T_0) \times \EXAX(\succeq)$.  If, in addition, $\succeq$ satisfies \Aone, then $\V$ can be taken to be a singleton (i.e., both $\P$ and $\U$ can be taken to be singletons).
\end{theorem}

\noindent Theorem \ref{thm:eurepa} is proved in the appendix.
The proof proceeds by first establishing a state-dependent
representation using $\mathcal{A}_0$ as the outcome space, and then, by changing the outcome space, `bootstrapping' the representation to an EU representation.  This technique shows that, when the state and outcome spaces are part of the representation, there is no difference between the formal requirement for a state-dependent representation and that for a SEU representation.  This does not mean that expected utility comes `for free'; rather, we interpret it to mean that the beliefs/desires formalism that motivates expected utility theory is sensible for the decision problems discussed in this subsection only if the particular outcome space chosen for the representation has some justification external to our theory. We note that if preferences satisfy \Aone, the theorem requires only the cancellation axiom $\AthreeAX$ rather than the stronger $\AthreeprAX$.

There are no uniqueness requirements on $\P$ or $\U$ in
Theorem~\ref{thm:eurepa}.  In part, this is because the state space and outcome space are not
uniquely determined.  But even if \Aone\ holds, so that the state space can be taken to be the set of atoms, the probability and the utility are far from unique, as the following example shows.

\xam\label{xam2}
Take $\A_0 = \{a,b\}$, $T_0 = \{t\}$, and $\T = \emptyset$.  Suppose that $\succeq$ is the reflexive transitive closure of the following string of preferences:
\begin{equation*}
a\succ\ift{t}{a}{b}\succ\ift{t}{b}{a}\succ b.
\end{equation*}
Every choice in $\A$ is equivalent to one of these four, so
\Aone\ holds, and we can take the state space to be $S^*=\{t,\neg
t\}$.  Let $O^* = \{o_1, o_2\}$, and define $\rho_{S^*O^*}^0$ so that 
$\rho_{S^*O^*}^0(a)$ is the constant function $o_1$ and
$\rho_{S^*O^*}^0(b)$ is the constant function $o_2$. Now define $\pi_{S^*}$ in the obvious way, so that $\pi_{S^*}(t)=\{t\}$ and  $\pi_{S^*}(\neg t)=\{\neg t\}$.  We can represent the preference relation
by using any probability measure $p^*$ such that $p^*(t) > 1/2$
and any utility function $u^*$ such that $u^*(o_1) > u^*(o_2)$.
\exam

As Example~\ref{xam2} shows, the problem really is that the set of
actions is not rich enough to determine the probability and utility.  By way of contrast, Savage's postulates ensure that the state space is
infinite and that there are at least two outcomes.  Since the acts are
all functions from states to outcomes, there must be uncountably many
acts in Savage's framework.

The next example shows that without the completeness axiom
\Aone, there may be no representation in which there is only one utility function and the state space is $\AtomsAX(T_0)$.
\xam\label{xam1}
Suppose that $T_0 = \emptyset$, $\T=\emptyset$, and $\A_0$ (and hence
$\A$) consists of the two primitive choices $a$ and $b$, which are
incomparable. In this case, the smallest state space that we can use has
cardinality at least 2. For if $|S|=1$, then there is only one
possible probability measure on $S$, 
so $a$ and $b$ cannot be incomparable.  Since there is only
one atom when there are no primitive propositions,  we cannot take the
state space to be the set of atoms. There is nothing special about
taking $T_0 = \emptyset$ here; similar examples can be constructed
showing that we cannot take the state space to be $\AtomsAX(T_0)$ for
arbitrary choices of $T_0$ if the preference relation is partial.  An easy
argument also shows that there is no representation where $|O|=1$.

This preference relation can be represented with two outcomes and two states.
Let $S=\{s_1,s_2\}$ and $O=\{o_1,o_2\}$.  Define
$\rho^0_{SO}(a)(s_i)=o_i$, and $\rho^0_{SO}(b)(s_i)=o_{2-i}$ for $i=
1,2$.  Let $\U$ contain a single function such that $u(o_1)\ne
u(o_2)$.  Let $\P$ be any set of probability measures including the
measures $p_1$ and $p_2$ such that $p_1(s_1)=1$ and $p_2(s_1)=0$.
Then the expected utility ranking of randomized acts under each $p_i$
contains no nontrivial indifference, and the ranking under $p_2$ is
the reverse of that under $p_1$.  Thus, these choices represent the
preference relation.

The assumption that there is only one utility function is critical
here.  For example, we could also represent this preference relation
using a single state and two
utility functions, $u_1$ and $u_2$, where $u_1(a) > u_1(b)$ and
$u_2(b) > u_2(a)$.  We conjecture that if we allow multiple utility
functions, then there is always a representation where the state space
is $\AtomsAX(T_0)$.
\exam

\subsection{A Representation Theorem for $\A^+$}
The purpose of this subsection is to show that for the language $\A^+$, we can get something much in the spirit of the standard representation theorem for AA acts.  The standard representation theorem has a mixture independence axiom and an Archimedean axiom.  As we have seen, $\AthreeprpAX$ gives us rational mixture independence; it does not suffice for full mixture independence.  To understand what we need, recall that the standard Archimedean axiom for AA acts has the following form:
\begin{arch}If $a\succ b\succ c$ then there exist $r,r'\in (0,1)$
such that $a\succ r a+(1-r)c\succ b \succ r' a+(1-r')c \succ c$.
\end{arch}
While this axiom is both necessary for and implied by the existence of
an SEU representation when $\succeq$ is complete, the following example
describes an incomplete preference relation $\succeq$ with a multi-probability  SEU
representation that fails to satisfy \Arch. 

\xam Suppose that $S = \{s_1,s_2,s_3\}$.
Let $a_1$, $a_2$, and $a_3$ be acts such that $a_i(s_j)$ gives an
outcome of 1 if $i=j$ and 0 otherwise.  Let $\P$ consist of all
probability distributions $p$ on $S$ such that $p(s_1) \ge p(s_2)
\ge p(s_3)$.  Define $\succeq$ by taking $a \succeq b$ iff the
expected outcome of $a$ is at least as large as that of $b$ with respect to all the probability distributions in $\P$.  It is easy to see that $a_1 \succ a_2 \succ a_3$, but for no $r \in (0,1)$ is it the case that $r a_1 +  (1-r_1) a_3 \succ a_2$ (consider the probability
distribution $p$ such that $p(s_1) = p(s_2) = 1/2$).
\exam

We can think of the Archimedean axiom as trying to capture some
continuity properties of $\succeq$.  We use instead the following
axiom, which was also used by \citet{Au62,au64}.  If the set
of tests has cardinality $n$ and the set of primitive choices has
cardinality $m$, we can identify an act $\A^+$ with an element of
$\Nspace{2^nm}$, so the graph of $\succeq$ (i.e., the set of pairs
$(x,y)$ such that $x \succeq y$) can be viewed as a subset 
of $\Nspace{2^{n+1}m}$.
\addtocounter{axiom}{1}
\begin{axiom}\label{Afour}
The graph of the preference relation $\succeq$ is closed.
\end{axiom}

As the following result shows, in the presence of $\AthreeprpAX$
(extended statewise mixture cancellation),
$\Afour$ implies full mixture independence.  Moreover, if we also
assume $\Aone$, then $\Afour$ implies $\Arch$.  Indeed, it will follow from 
Theorem \ref{thm:aaeurep}
that in the presence of \Aone\ and $\AthreeprpAX$, \Afour\ and \Arch\ are equivalent.  On the other hand, it seems that \Arch\ does not suffice to capture independence if $\succeq$ is a partial order.  Summarizing, \Afour\ captures the essential features of the Archimedean property, while being more appropriate if $\succeq$ is only a partial order.

\begin{prop}\label{prop:arch}
(a) $\AthreeprpAX$ and $\Afour$ imply full mixture independence.\\
(b) \Aone, $\Afour$, and \ecm\ together
imply \Arch.
\end{prop}
\medskip

\begin{definition}
A preference relation $\succeq$ on a set $C \subseteq\A^+_{\A_0,T_0}$ has a \emph{constructive $\T$-consistent
SEU representation} iff there is a finite set of states $S$, a finite set $O$ of outcomes, a set $\U$ of utility functions $u:O\to\R$, a
closed set $\P$ of probability distributions on $S$, a
closed set $\V\subseteq \U\times\P$, a test interpretation $\pi_S$
consistent with $\T$, and a choice interpretation $\rho_{SO}$ such that
$a\succeq b$ iff
\begin{equation*}
\sum_{s,o} u(o)\rho_{SO}(a)(s)(o) p(s)\geq
\sum_{s,o} u(o)\rho_{SO}(b)(s)(o) p(s)\
\mbox{for all $(u,p) \in \V$}.
\end{equation*}
\end{definition}
In the statement of the theorem, if $\succeq$ is a preference relation
on a mixture-closed subset of $\A^+_{\A_0,T_0}$, we use
$\succeq \otimes \, (a,b)$ to denote the smallest preference relation
including $\succeq$ and $(a,b)$ satisfying $\AthreeprpAX$ and $\Afour$, and take $\EXAXp(\succeq)$ to consist of all complete preference relations extending $\succeq$ and satisfying $\AthreeprpAX$ and $\Afour$.

\begin{theorem}\label{thm:aaeurep}
A preference relation $\succeq$ on a closed and mixture-closed set $C \subseteq \A^+_{\A_0,T_0}$ has a constructive $\T$-consistent SEU representation iff $\succeq$ satisfies $\AthreeprpAX$ and $\Afour$.  Moreover, in the representation, either $\P$ or $\U$ can be taken to be a singleton and, if $\U$ is a singleton $\{u\}$, the state space can be taken to be $\AtomsAX(T_0) \times \EXAXp(\succeq)$.  If, in addition, $\succeq$ satisfies \Aone, then $\V$ can be taken to be a singleton.
\end{theorem}

As in the case of the language $\A$, we cannot in general take the state space to be the set of atoms.  Specifically, if $\A_0$ consists of two primitive choices and we take all choices in $\A_0^+$ to be incomparable, then the same argument as in Example~\ref{xam1} shows that we cannot take $S$ to be $\Atoms(T_0)$, and there are no interesting uniqueness requirements that we can place on the set of probability measures or the utility function.  On the other hand, if \Aone\ holds, the proof of Theorem~\ref{thm:aaeurep} shows that, in the representation, the expected utility is unique up to affine transformations.  That is, if $(S,O,p,\pi_{S},\rho_{SO}^0,u)$ and $(S',O',p',\pi_{S'},\rho_{S'O'}^0,u')$ are both representations of $\succeq$, then there exist constants $\alpha$ and $\beta$ such that for all acts $a \in \A_{\A_0,T_0}^+$, $\mathrm{E}_{p}(u({\rho_{SO}(a)})) = \alpha\mathrm{E}_{p'}(u'({\rho_{S'O'}(a)})) +  \beta$.

\subsection{Objective Outcomes}\label{sec:objectiveoutcomes}
In choosing, for instance, certain kinds of insurance or financial
assets, there is a
natural, or objective, outcome space---in these cases, monetary
payouts.
To model this, we take the set $O$ of objective outcomes as given, and
identify it with a subset of the primitive acts $A_0$.  Call the
languages with this distinguished set of outcomes $\A_{\A_0,T_0,O}$ and
$\A^+_{\A_0,T_0,O}$, depending on whether we allow randomization.

To get a representation theorem in this setting, we need to make some
standard assumptions.  The first is that there is a best and worst
outcome; the second is a state-independence assumption.  However, this
state-independence assumption only applies to acts in $O$,
but not to all acts.

\begin{axiom}\label{Afive}There are outcomes $o_1$ and $o_0$ such that
for all non-null tests $t$, $o_1\succeq_t a\succeq_t o_0$ for all $a\in\A_0$.
\end{axiom}
\begin{axiom}
\label{Asix} If $t$ is not null and $o, o' \in O$, then
$o \succeq o'$ iff $o \succeq_t o'$.
\end{axiom}

In all our earlier representation theorems, it was possible to use a
single utility function.  $\Afive$ and $\Asix$ do not suffice to get such a representation theorem.  A necessary condition to have a single utility function, if we also want utility to be state independent, is that $\succeq$ restricted to $O$ be complete.

\begin{axiom}\label{Aseven}
$\succeq$ restricted to $O$ is complete.
\end{axiom}

While $\Asix$ and $\Aseven$ are necessary to get a representation with a single utility function, they are not sufficient, as the following
example shows.

\xam\label{xam:singleutility}
Suppose that we have a language with one primitive test $t$, and three
outcomes, $o_0$, $o_1$, and $o$.  Let $a_1$ be $\ift{t}{o_0}{o_1}$
and let $a_2$ be $\mathbf{if}\ t \ \mathbf{then}$ $o_1\ \mathbf{else}\ o_0$.  Let $\succeq$ be the smallest preference relation satisfying $\AthreeprAX$, $\Afive$, $\Asix$, and $\Aseven$  (or $\AthreeprpAX$, $\Afour$, $\Afive$, $\Asix$, and $\Aseven$, if we are considering the language $\A^+$) such that $o \sim a_1$
and $o_1\succ o_0$.
Note that $a_1$ and $a_2$ are incomparable according to $\succeq$.  Suppose that there were a representation of $\succeq$ involving a set $\P$ of probability measures and a single utility function $u$.  Thus, there would have to be probability measures $p_1$ and $p_2$ in $\P$ such that $a_1 \succ a_2$ according to $(p_1,u)$ and $a_2 \succ a_1$ according to $(p_2,u)$.  It easily follows that $p_1(\pi_S(t)) < 1/2$ and  $p_2(\pi_S(t)) > 1/2$.  We can assume without loss of generality (by using an appropriate affine transformation) that $u(o_0) = 0$ and $u(o_1)=1$.  Since $o \sim a_1$, $u(o)$ must be the same as the expected utility of~$a_1$.  But this expected utility is less than $1/2$ with $p_1$ and more than $1/2$ with $p_2$.  This gives the desired contradiction.
\exam

Part of the problem is that it is not just the acts in $O$ that must be state independent.   Let $O^+$ be the smallest set of acts containing $O$ that is closed under convex combinations, so that if $o$ and $o'$ are in $O^+$, then so is $ro+(1-r)o'$.  Let $\Asix^+$ and $\Aseven^+$ be the axioms that result by replacing $O$ by $O^+$ in $\Asix$ and $\Aseven$, respectively. Example~\ref{xam:singleutility} actually shows that $\Asix^+$ and $\Aseven$ do not suffice to get a single utility function; Theorem~\ref{thm:outcomes} shows that $\Asix^+$ and $\Aseven^+$ do, at least for $\A^+$.  We do not have a representation theorem for $\A$, and believe it will be hard to obtain such a theorem (for much the same reasons that it is hard to get a representation theorem in the Savage setting if we restrict to a finite set of acts).

\begin{theorem}\label{thm:outcomes}
A preference relation $\succeq$ on a set $C \subseteq \A^+_{\A_0,T_0,O}$ has a constructive SEU representation with outcome space $O$ iff $\succeq$ satisfies $\AthreeprpAX$, $\Afour$, $\Afive$, $\Asix^+$, and $\Aseven^+$.  Moreover, in the representation, $\U$ can be taken to be a singleton $\{u\}$ and the state space can be taken to
have the form $\AtomsAX(T_0) \times \A_0 \times \EXAXp(\succeq)$.
If in addition $\succeq$ satisfies \Aone, then we can take $\V$ to be a singleton too.
\end{theorem}

Note that, even if $\succeq$ satisfies $\Aone$, the state space
has the form $\AtomsAX(T_0) \times \A_0$.  The fact that we cannot take the state space to be $\AtomsAX(T_0)$ is a consequence of our assumption that primitive acts are deterministic.  Roughly speaking, we need the extra information in states to describe our uncertainty regarding how the primitive acts not in $O$ can be viewed as functions from states to outcomes in the pre-specified set $O$.  The following example shows that we need the state space to be larger than $\AtomsAX(T_0)$ in general.

\xam\label{xam3}
Suppose that there are no primitive propositions, so $\AtomsAX(T_0)$ is a singleton.  There are three primitive acts in $\A_0$: $o_1=\$50000$, $o_0=\$0$, and $a$, which is interpreted as buying 100 shares of Alphabet.  Suppose that $o_1\succ a\succ o_0$.  If there were a representation with only one state, then $\rho_{SO}(a)$ would have to be either $o_1$ or $o_0$, which would imply that $a\sim o_1$ or $a\sim o_0$, contradicting our description of $\succeq$.  The issue here is our requirement that a primitive choice be represented as a function from states to outcomes.  If we could represent $a$ as a lottery, there would be no problem representing $\succeq$ with one state.  We could simply take $u(o_1) = 1$, $u(o_0) = 0$, and take $a$ to be a lottery that gives each of $o_0$ and $o_1$ with probability $1/2$.

We prefer not to allow the DM to consider such `subjective' lotteries.
Rather, we have restricted the randomization to acts.
The representation theorem would not change if we allowed primitive
choice to map a state to a distribution over outcomes, rather than
requiring them to be mappings from states to single outcomes (except
that we could take the state space to be $\AtomsAX(T_0)$).

We can easily represent $\succeq$ using two states, $s_0$ and $s_1$, by taking $a$ to be the act with outcome $o_i$ in state $s_i$, for $i = 0,1$.  Taking each of $s_0$ and $s_1$ to hold with probability $1/2$ then gives a representation of $\succeq$.  In this representation,
we can view $s_0$ as the state where buying Alphabet is a good investment, and $s_1$ as the state where buying Alphabet is a bad investment.  However, the DM cannot talk about Alphabet being a good investment; this is not part of his language.  Another DM might explicitly consider the test that Alphabet is a good investment.  Suppose this DM has the same preference relation over primitive acts as in the example, is indifferent between 100 shares of Alphabet and $o_1$ if the test is true, and is indifferent between 100 shares of Alphabet and $o_0$ if the test is false. This DM's preference relation has
exactly the same representation as the first DM's preference relation, but
now $s_0$ can be viewed as the atom where the test is false, and $s_1$
can be viewed as the atom where the test is true.  The second DM can
reason about Alphabet being a good investment explicitly, and can talk to others about it.
\exam

\section{Nonstandard Interpretations}\label{sec:nonstandard}

Failures of extensionality are concerned with a semantic issue:
Identifying when two descriptions of an event are equivalent.
However, the same issue arises with respect to logical
equivalence. For instance, suppose act $a$ gives a DM $x$ if
$\lnot((\lnot 
t_1\land\lnot t_3)\lor(\lnot t_2\land\lnot t_3))$ is true, and $y$
otherwise, and act $b$ gives the DM $x$ if $(t_1\land t_2)\lor t_3$ is
true, and $y$ otherwise.  Only someone adept at formula manipulation
(or with a good validity-checking program) will recognize that acts $a$
and $b$ are equivalent as a matter of logic because the equivalence of
the two compound propositions is a tautology.  So far we have required 
that DMs be \emph{logically omniscient}, and recognize all such
tautologies, because   we have considered only standard
interpretations.

An interpretation $\pi_S$ on a state space $S$ does not have to be
standard; all that is required is that it associate with each
test a subset of $S$.  By allowing nonstandard test interpretations, we can back off from our requirement that DM's know all tautologies.
This gives us a way of modeling failures of logical omniscience and, in
particular, \emph{resource-bounded reasoning} by DMs.
We remark that such nonstandard
interpretations are essentially what philosophers call `impossible
possible worlds' \citep{Rant}; they have also been used in game theory
for modeling resource-bounded reasoning \citep{Lip99}.

A standard test interpretation is completely determined by its behavior on the primitive tests.  However, in general, there is no
similar finite characterization of a nonstandard test interpretation.
To keep things finite, when dealing with nonstandard interpretations,
we assume that there is a finite subset $T^*$ of the set $T$ of all
tests such that the only tests that appear in choices are those in
$T^*$.  (This is one way to model resource-bounded reasoning.)
With this constraint, it suffices to consider the behavior of a nonstandard interpretations only on the tests in $T^*$.  Let $\A_{\A_0,T^*}$ consist of all choices whose primitive choices are in $\A_0$ and whose tests are all in $T^*$.

The restriction to choices in $\A_{\A_0,T^*}$ allows us to  define the
cancellation postulate in a straightforward way even in the presence of
nonstandard interpretations.  A \emph{truth assignment to $T^*$} is just
a function $v: T^* \rightarrow \{\true,\false\}$.  We can identify an
interpretation on $S$ with a function that associates with every state
$s \in S$ the truth assignment $v_s$ such that $v_s(t) = \true$ iff $s
\in \pi(t)$.  For a standard interpretation, we can use an atom instead
of a truth assignment, since for a standard interpretation, the behavior
of each truth assignment is determined by its behavior on primitive
propositions, and we can associate with the truth assignment $v_s$ that
atom $\delta_s$ such that $t$ is a conjunct in $\delta_s$ iff $v_s(t) =
\true$.  These observations suggest that we can consider truth
assignments to be the generalization of atoms once we move to nonstandard
interpretations.   Indeed, if we do this, we can easily generalize all
our earlier theorems.

In more detail, we now view a choice as a function, not from atoms to
primitive choices, but, more generally, as a function from truth
assignments to primitive choices.  As before, we take primitive 
choices to be constant functions.  The choice
$a=\ift{t_1}{a_1}{(\ift{t_2}{a_2}{a_3})}$
can be identified with the function $f_a$ such that
\begin{equation*}
f_a(v) = \begin{cases}
a_1 &\mbox{if $v(t_1) = \true$}\\
f_{\scriptscriptstyle\ift{t_2}{a_2}{a_3}} &\mbox{if $v(t_1) = \false$}
\end{cases}
\end{equation*}
and
\begin{equation*}
f_{\scriptscriptstyle\ift{t_2}{a_2}{a_3}} =
\begin{cases}
a_2 &\mbox{if $v(t_2) = \true$}\\
a_3 &\mbox{if $v(t_2) = \false$.}
\end{cases}
\end{equation*}
A truth assignment $v$ is consistent with $\T$ if
$v(t) = \true$ for all tests $t \in \T$.

With these definitions in hand, all our earlier results hold, with the following changes:
\begin{itemize}
\item we replace $\A_{\A_0,T_0}$ by $\A_{\A_0,T^*}$;
\item we replace `atoms $\delta$ over $T_0$' by `truth assignment to
$T^*$'.
\end{itemize}
The cancellation axioms are all now well defined.  With these changes,
Proposition~\ref{prop:equivalence} and
Theorems~\ref{thm:eurepa},~\ref{thm:aaeurep}, and~\ref{thm:outcomes}
hold with essentially no changes in the proof.  Thus, we have
representation theorems that apply even to resource-bounded
reasoners.

There is one further subtlety, however.  A theory puts
constraints on the set of test interpretations we  consider.
Up to now, we have taken a theory $\T$ to be a collection of tests.
If we restrict to standard interpretations, this suffices, in the
sense that, given a state space $S$, for all sets $\I$ of standard
interpretations for the state space $S$ of a set $T$ of 
tests, there is a theory $\T_{\I}$ such that a test
interpretation $\pi_S$ for $S$  respects $\T_{\I}$ iff $\pi_S
\in \I$.  That is, a theory can specify a set of interpretations.
This is no longer the case once we move to nonstandard
interpretations.  For example, the nonstandard interpretation $\pi_S^*$
for $S$ that
makes every test true at every state (i.e., 
$\pi_S^*(t) = S$ for all tests $t$) respects every theory.  There
is no way that a theory can disallow $\pi_S^*$.

To understand why this is an issue, note that with standard
interpretations, we get many properties for free, so to speak.  For
example, for all tests $t$, $t$ and $\neg t$ cannot be simultaneously true.
More precisely, for all standard
interpretations $\pi_S$ for $S$, we have $\pi_S(t) \cap \pi_S(\neg t) =
\emptyset$.  We might want to restrict to nonstandard interpretations
that have this property.  Unfortunately, there is no way to do this if a
theory is just a set of tests; we cannot exclude the
interpretation $\pi_S^*$, and it does not have this property.
Similarly, we may want to restrict to interpretations where
conjunction is commutative, so $\pi_S(t_1 \land t_2) = \pi_S(t_2 \land
t_1)$.  Note that just adding the test $(t_1 \land t_2) \dimp (t_2
\land t_1)$ does not have this effect; it just ensures that $\pi_S((t_1
\land t_2) \dimp (t_2 \land t_1)) = S$.

We thus consider \emph{generalized theories}, which consist of
\emph{generalized axioms} of the form $T_1 = T_2$, where $T_1$ and $T_2$ 
are sets of tests.
We also assume that there are special primitive
tests $\btrue$ and $\bfalse$.  
A test interpretation $\pi_S$ for a state space $S$ respects a generalized
theory $\T$ if, for every generalized axiom $T_1 = T_2$ in $\T$,
$\{s \in S: s \in \pi_S(t) \mbox{ for all } t \in T_1\} =
\{s \in S: s \in \pi_S(t) \mbox{ for all } t \in T_2\}$; that is, we want all 
the tests in $T_1$ to be true at a state $s$ iff all the tests in
$T_2$ are true.  Moreover, we require that $\pi_S(\btrue) = S$ and
$\pi_S(\bfalse) = \emptyset$.  


Working with generalized theories gives us a great deal more power to
put constraints on nonstandard interpretations.  For example: 
\begin{itemize}
\item If $\T$ contains the generalized axiom $T = \{\btrue\}$, then
  all the tests in $T$ 
  are true at every state of a nonstandard interpretation that
  respects $\T$.  (Thus, this particular generalized axiom has the
  same effect as 
  the theories we considered earlier.)
\item If $\T$ contains
$\{t_1 \land t_2\} = \{t_2 \land t_1\}$ for all tests $t_1$ and $t_2$, then
    conjunction is commutative for all nonstandard interpretations
    that respect $\T$.
\item If $\T$ contains $\{t, \neg t\} =
\{F\}$, then at most one of $t$ and $\neg t$ will be true according to
$\pi_S$.
\end{itemize}

Despite their added expressive power (once we allow nonstandard
interpretations), Proposition~\ref{prop:equivalence} and
Theorems~\ref{thm:eurepa},~\ref{thm:aaeurep}, and~\ref{thm:outcomes}
continue to hold with essentially no changes in the proof if we allow
generalized theories.

\section{Updating}\label{sec:updating}
There is nothing unique about the state space chosen for an SEU
representation of a given choice problem.  Our representation theorems
state that if an SEU representation exists on any given state space and outcome space with test and choice interpretation functions, then
preferences satisfy the appropriate cancellation and other appropriate
axioms.  Our proofs, however, show that
(for standard interpretations)
we can always represent a choice situation on the state space
$\AtomsAX(T_0)\times\EXAX(\succeq)$, or
$\AtomsAX(T_0)\times\A_0\times\EXAX(\succeq)$ for the
objective-outcomes case, so this construction is in some sense
canonical.
Just as importantly, this state space respects the DM's choice of
language and her preferences, and seems like a natural state space for
the DM to use when reasoning about the decision problem.

In our models, there are two kinds of information.  A DM can learn more
about the external world, that is, learn the results of some tests.  A
DM can also learn more about her internal world, that is, she can learn
more about her preferences.  This learning takes the form of adding more
comparisons to her (incomplete) preference relation.  To make this precise,
given a preference relation $\succeq$ on a set $C \subseteq \A_{\A_0,T_0}$
satisfying $\AthreeprAX$, let $\succeq \oplus\,(a,b)$ be the smallest
preference relation including $\succeq$ and $(a,b)$ satisfying
$\AthreeprAX$.  (There is such a smallest preference relation, since it is
easy to see that if $\succeq'$ and $\succeq''$ both extend $\succeq$,
include $(a,b)$, and satisfy $\AthreeprAX$, then so does
$\succeq'\cap\succeq''$.)

If we take the state space to be $\AtomsAX(T_0) \times
\EXAX(\succeq)$, then a DM's preference relation after obtaining either
new test information or new comparison information can be represented
by conditioning the original probability measures. 
If $\P$ is a set of probability distributions on some set $S$ and $E$ is
a measurable subset of $S$, let $\mbox{$\P\mid E$} = \{q:q=p(\,\cdot\,|E)\
\mbox{for some $p\in\P$ with $p(E)>0$}\}$.
That is, in computing $\P \mid E$, we throw out all
distributions $p$ such that $p(E) = 0$, and then apply standard
conditioning to the rest.  Let $\P\mid t=\P\mid\pi_S{(t)}$.  In the
theorems below, we condition on a test $t$ and on a partial order
$\succeq'$ extending $\succeq$.  We are implicitly identifying $t$ with the event $\{\delta \in \AtomsAX(T_0):\delta\rimp t\}$, and $\succeq'$ with the set of total orders in $\EXAX(\succeq)$ extending $\succeq$.

\newtheorem*{eurepac}{Theorem \ref{thm:eurepa}c}
\begin{eurepac}Under the assumptions of Theorem \ref{thm:eurepa}, and
with a representation of $\succeq$ in which $S=\AtomsAX(T_0)\times\EXAX(\succeq)$ and
$\U$ is a singleton $\{u\}$, $\succeq_t$ is represented by $\P \mid t$
and $u$, and $\succeq \oplus \, (a,b)$ is represented by
$\P \mid (\succeq \, \oplus \, (a,b))$ and~$u$.
\end{eurepac}

\newtheorem*{aaeurepc}{Theorem \ref{thm:aaeurep}c}
\begin{aaeurepc} Under the assumptions of Theorem \ref{thm:aaeurep}, and
with a representation of $\succeq$ in which
$S=\AtomsAX(T_0)\times\EXAXp(\succeq)$ and $\U$ is a singleton $\{u\}$,
$\succeq_t$ is represented by $\P \mid t$ and $\{u\}$, and $\succeq
\otimes\, (a,b)$ is represented by $\P\mid(\succeq\otimes\,(a,b))$
and~$u$.
\end{aaeurepc}

\newtheorem*{outcomesc}{Theorem \ref{thm:outcomes}c}
\begin{outcomesc} Under the assumptions of Theorem \ref{thm:outcomes},
and with a representation of $\succeq$ in which $S=\AtomsAX(T_0)\times\A_0
\times\EXAXp(\succeq)$ and $\U$ is a singleton $\{u\}$,
$\succeq_t$ is represented by $\P \mid t$ and $u$, and
$\succeq \otimes \, (a,b)$ is represented by
$\P \mid (\succeq \, \otimes \, (a,b))$ and~$u$.
\end{outcomesc}

Information in the external world is modeled as a restriction on the set of feasible acts; information in the internal world is adding
comparisons to a the preference relation.  These theorems show that both kinds of information can be modeled within a Bayesian paradigm.

\section{Conclusion}\label{sec:conclusion}
Our formulation of decision problems has several advantages over more
traditional formulations. First, we theorize about only the actual observable choices available to the DM, without having states and outcomes, and without needing to view
choices as functions from states to outcomes. In contrast, in many decision theory
experiments, when the DM is given a word problem, the experimenter has
an interpretation of this problem as a choice among Savage acts. The
experimenter is then really testing whether the DM's choices are
consistent with 
decision theory \emph{given this interpretation}.  

Thus, a joint hypothesis is being tested.  Standard decision theory can be
rejected only if the other part of the joint test---that the
experimenter and DM interpret represent the word problem with identical
Savage acts---is maintained as true.  
As we have shown, given the decision maker's theory, our approach
places restrictions on choices. Thus, although our approach also leads
to a joint hypothesis, we have a framework for reasoning about what the
decision maker knows about the world (a theory) and about what
restrictions on choices arise from SEU given the decision maker's
theory, rather than the experimenter's theory.  
In fact, if an experimenter is willing to take the existence of an SEU representation
as a maintained hypothesis, she can test hypotheses about the DMs theory of the world.

Second, by viewing choices as \emph{syntactic} objects, our approach
allows us to consider DMs who associate different meanings to the same object of choice.  Moreover, that meaning can depend on the DM's theory of the world.  A DM might have a theory that does not recognize
equivalences between certain tests, and thus choices, that may be
obvious to others. This potential difference between a DM's theory of
the world and an experimenter's view of the world provides an
explanation for framing effects, while still allowing us to view a DM
as an expected utility maximizer.  Moreover, since a DM's theory may
not contain all of standard propositional logic, we can model
resource-bounded DMs who cannot discern all the logical consequences
of their choices.  The existence of an SEU representation and the
presence of framing effects are independent once one is free to
choose a state space.

Third, our approach allows us to consider different DMs who use different languages to describe the same phenomena.  To see why this might be important, consider two decision makers who are interested in $100$ shares of Alphabet stock and money (as in Example~\ref{xam3}).
Suppose that one DM considers quantitative issues like the 
price/earnings ratio to be relevant to the future value of Alphabet, while the other considers astrological tables relevant to Alphabet's future value.  The DM who uses astrology might not understand price/earnings ratios (the notion is simply not in his vocabulary) and, similarly, the DM who uses quantitative methods might not understand what it means for the moon to be in the seventh house.  Nevertheless, they can trade Alphabet stock and money, as long as they both have available primitive actions like `buy 100 shares of Alphabet' and `sell 100 shares of Alphabet'. If we model these decision problems in the Savage framework, we would have to think of assets as Savage acts on a common state and outcome space.  Our approach does not require us to tie the DM's decision problems together with a common state space.
Every DM acts as if she has a state space, but these state spaces may
be different. Even if agents agree on the formulas of interest, they
may interpret them completely differently.  Thus, agreeing to
disagree results \cite{Au}, which say  
that DMs with a common prior must have a common posterior (they
cannot agree to disagree) will in general not hold.

A fourth advantage of our approach is more subtle, but potentially profound.  Representation theorems are just that; they merely provide a description of a preference relation in terms of numerical scales.  Decision theorists make no pretense that these representations have anything to do with the cognitive processes by which individuals make decisions.  But to the extent that the language of choices models the language of the DM, we have the ability to interpret the effects of cognitive limitations having to do with the language in terms of the representation.  Our approach allows us to consider the possibility that there may be limitations on the space of choices because some sequence of tests is too computationally costly to verify. Our model of nonstandard test interpretations also takes into account a DM's potential inability to recognize that two choices logically represent the same function.

There is clearly more than could be done to develop our approach
and apply it in various settings. One obvious development would be to
consider decision making over time, which will require us to consider
learning in more detail.  Learning in our 
framework is not just a matter of conditioning, but also learning about
new notions (i.e., becoming aware of new tests).
Note that considering dynamic decision-making will require us to take a
richer collection of objects of choice, a programming language that
allows (among other things) sequential actions (do this, then do that,
then do that). A second direction to consider is multi-agent decision
making.  As we have suggested in examples, once we move to a
multi-agent case, we can consider agents who may use
different languages.   There is clearly a connection here between our
framework and the burgeoning literature on awareness and its
applications to game theory (see, for example, 
\cite{Feinberg04,FH,Hal34,HR06,HMS03,karni2013reverse,MR99} that needs
to be explored.

We have shown how our approach can model framing problems as a consequence of the agent having a different theory from the modeler.  We believe that our approach can also model other `deviations' from rationality
of the type reported by  \citet{Luce90}, such as failures of
particular accounting identities that lead to the requirement that a
DM is indifferent between formally equivalent gambles.  These
deviations can be 
viewed as a consequence of an agent's bounded processing power.  This
will require us to be able to distinguish sentences such as, for
example, $2/3 (1/4 a + 3/4 b) + 1/3 c$ and $1/6 a + 5/6 (3/5 b + 2/5 c)$.
To do that, we need to give semantics to choices that does not view
them as functions from states to distributions over outcomes. 
We leave further exploration of all these issues to future work.


\section*{Appendix}

\begin{proof}[\textbf{\upshape Proof of Proposition
\ref{prop:cancel}}]
First suppose that cancellation holds.  To see that $\succeq$ is
reflexive, take $n=1$ and $a_1 = b_1 = a$ in the cancellation axiom.
The hypothesis of the cancellation axiom clearly holds, so we must have $a \succeq a$. To see that cancellation implies transitivity, consider the pair of sequences $\langle a,b,c\rangle$ and $\langle b,c,a\rangle$.  Cancellation clearly applies.  If $a\succeq b$ and $b\succeq c$, then cancellation implies $a\succeq c$.  

We also need to prove the converse part of the
proposition.  For the converse, suppose that $\succeq$ is reflexive and transitive. By way of contradiction, suppose that $\langle a_1, \ldots, a_n\rangle$ and $\langle b_1, \ldots, b_n\rangle$ are two sequences of minimal cardinality $n$ that violate cancellation; that is, $\{\{a_1, \ldots, a_n\}\}=\{\{b_1, \ldots, b_n\}\}$,
$a_i\succeq b_i$ for $i \in \{1, \ldots, n-1\}$, and it is not the case that $b_n\succeq a_n$

If $n=1$, and $\{\{a\}\} = \{\{b\}\}$, then we must have $a=b$,
and the cancellation postulate holds iff $a \succeq a$, which follows
from our assumption that $\succeq$ is reflexive.

If $n > 1$, since the two multisets are equal, there must be some
permutation $\tau$ of $\{1,\ldots, n\}$ such that $a_{\tau(i)} = b_i$.
Let $\tau^j(1)$ be the result of applying $\tau$ $j$ times, beginning with $1$.  Let $k$ be  the first integer such that $\tau^{k+1}(1)$ is either 1 or~$n$. Then we have the  situation described by the following table, where the diagonal arrow denotes equality.
\begin{center}\begin{tabular}{ccccc}
$a_1$&$\succeq$&$b_1$\\
    &$\swarrow$\\
$a_{\tau(1)}$&$\succeq$&$b_{\tau(1)}$\\
    &$\swarrow$\\
$\vdots$&&$\vdots$\\
    &$\swarrow$\\
$a_{\tau^k(1)}$&$\succeq$&$b_{\tau^k(1)}$
\end{tabular}\end{center}
Note that we must have $k \le n-1$.  If $\tau^{k+1}(1) = 1$, then
$b_{\tau^{k}(i)} = a_1$.  Thus, the multisets $\{\{a_1,\ldots,a_{\tau^k(1)}\}\}$ and $\{\{b_1,\ldots,b_{\tau^k(1)}\}\}$ must be equal.
The sequences that remain after removing
$\{\{a_1,\ldots,a_{\tau^k(1)}\}\}$ from the first sequence and
$\{\{b_1,\ldots,b_{\tau^k(1)}\}\}$ from the second also provide
a counterexample to the cancellation axiom, contradicting the minimality of $n$. Thus, $\tau^{k+1}(1)=n$, and we can conclude by transitivity that $a_1\succeq a_n$.

Continuing on with the iteration procedure starting with $a_{\tau^{k+2}(1)} = a_{\tau(n)} = b_n$, we ultimately must return to $a_1$ and $b_1$, as illustrated in the following table: $a_1$ and $b_1$.
\begin{center}\begin{tabular}{ccccc}
$b_n=a_{\tau^{k+1}(1)}$&$\succeq$&$b_{\tau^{k+1}(1)}$\\
    &$\swarrow$\\
$\phantom{b_n\rightarrow}\vdots$&&$\vdots$\\
    &$\swarrow$\\
$\phantom{b_n\rightarrow} a_{\tau^l(1)}$&$\succeq$&$b_{\tau^1(1)}\rightarrow a_1\mkern-45mu$
\end{tabular}\end{center}
It follows from transitivity that $b_n\succeq a_1$.  By another application of transitivity, we conclude that $b_n\succeq a_n$.  This contradicts the hypothesis that the original sequence violated the cancellation axiom.
\end{proof}

\begin{proof}[\textbf{\upshape Proof of Proposition
      \ref{prop:cancellationbound}}]
Let $PO(C)$ denote all the preference relations $\succeq$ on $C$.  Since
$C$ is finite, so is $PO(C)$. 
Let $PO(C,n)$ denote the subset of $PO(C)$ consisting of all
preference relations $\succeq$ satisfying $SC_1, \ldots, SC_n$.
$PO(C,1), PO(C,2), \ldots$ is clearly a nondecreasing sequence of sets
of preference relations.  Since $PO(C)$ is finite, the sequence must
stabilize at some point; that is, there must exist some $N$ such that
$PO(C,n) = PO(C,N)$ for all $n \ge N$.  It follows that if a
preference relation $\succeq$ on $C$ satisfies $SC_1, \ldots, SC_N$, then
it satisfies $SC_n$ for all $n$.
\end{proof}

\begin{proof}[\textbf{\upshape Proof of Proposition
\ref{prop:savagecancel}}]
Take $\<a_1,a_2\>=\<a_Tc,b_Tc'\>$ and take $\<b_1,b_2\> = \<b_T c,
a_T c'\>$.  Note that for each state $s \in T$, $\{\{a_Tc(s),b_T c'(s)
\}\}=\{\{a(s),b(s)\}\}=\\
\{\{b_T c(s),a_T c'(s)\}\}$, and for each state $s \notin T$, $\{\{a_T c(s),b_T c'(s)\}\}=\\
\{\{c(s),c'(s)\}\}=\{\{b_T c(s),a_T c'(s)\}\}$.  Thus, we can apply statewise cancellation to get that if $a_Tc \succeq b_Tc$, then $a_Tc' \succeq b_Tc'$.
\end{proof}

\begin{proof}[\textbf{\upshape Proof of Proposition
\ref{prop:cancelequivalent}}]
Suppose the hypotheses of extended statewise cancellation hold.  If $b_n \succeq a_n$, we are done.  If not, by \Aone, $a_n\succeq b_n$.  But then the hypotheses of statewise cancellation hold, so again, $b_n \succeq a_n$.
\end{proof}

\begin{proof}[\textbf{\upshape Proof of Theorem
\ref{thm:cancellationchar1}}]
Suppose that $\succeq$ satisfies \ecm.  Then it satisfies \ca, and so from Proposition~1, $\succeq$ is reflexive and transitive.  To show that $\succeq$ satisfies rational mixture independence, suppose that $a \succeq b$ and $r = m/n$.  Let $a_1=\cdots = a_m = a$ and $a_{m+1} = \cdots = a_{m+n} = rb + (1-r)c$; let $b_1 = \cdots = b_m = b$ and $b_{m+1} = \cdots = b_{m+n} = ra + (1-r)c$. Then $\sum_{i=1}^{m+n}a_i = \sum_{i=1}^{m+n}b_i$, and so $ra+(1-r)c\succeq rb+(1-r)c$.

Similarly, if $ra+(1-r)c\succeq rb+(1-r)c$, then applying \ecm\ to the same sequence of acts shows that $a \succeq b$.

For the converse, suppose that $\succeq$ is reflexive, transitive, and
satisfies rational mixture independence. Suppose that
$\<a_1 ,\ldots,  a_n\>$ and $\< b_1,  \ldots,  b_n\>$ are sequences of
of elements of $C$ such that
$a_1 + \cdots + a_n = b_1 + \cdots + b_n$, $a_i \succeq b_i$ for $i = 1,
\ldots, n-k$, $a_{k+1} = \ldots = a_n$, and $b_{k+1} = \ldots = b_n$.
Then from transitivity and rational mixture independence we get that
\begin{multline*}
\frac{1}{n}(a_1 + \cdots + a_n)\succeq
\frac{1}{n}(b_1 + \cdots +
b_{k} + a_{k+1} + \cdots + a_{n})\\
=\frac{1}{n}(b_1 + \cdots + b_k) +
\frac{n-k}{n}a_n.
\end{multline*}
Since $b_{k+1} = \ldots = b_n$ and $a_1 + \cdots + a_n = b_1 + \cdots +b_n$, we have that
\begin{equation*}
\frac{1}{n}(b_1 + \cdots + b_k) +\frac{n-k}{n}(b_n)=
\frac{1}{n}(b_1 + \cdots + b_n)=
\frac{1}{n}(a_1 + \cdots + a_n).
\end{equation*}
Thus, by transitivity,
$$
\shortv{\begin{array}{ll}
&}
\frac{1}{n}(b_1 + \cdots + b_k) +\frac{n-k}{n}(b_n)
\fullv{\succeq}
\shortv{\\ \succeq &}
\frac{1}{n}(b_1 + \cdots + b_k) +\frac{n-k}{n}(a_n).
\shortv{\end{array}}
$$
By rational mixture  independence, it follows that \mbox{$b_n \succeq a_n$}.
\end{proof}

\begin{proof}[\textbf{\upshape Proof of Proposition
\ref{prop:equivalence}}]
Let $S=\AtomsAX(T_0)$, the set of atoms consistent with
$\T$, let $O$ be $\A_0$, the set of primitive choices, and define
$\rho_{SO}^0(c)$ to be the constant function $c$ for a primitive
choice $c$.  It is easy to see that $\rho_{SO}(c) = f_c$ for all
choices $c$.  If $a\equiv_{\T}b$, then $\rho_{SO}(a)=\rho_{SO}(b)$, so we must have $f_a = f_b$.  Now apply $\AthreeAX$ with $a_1=a$ and $b_1=b$ to get $b\succeq a$, and then reverse the roles of $a$ and $b$.
\end{proof}

\begin{proof}[\textbf{\upshape Proof of Proposition
\ref{prop:arch}}]
 For part (a), suppose that $a \succeq b$, and
$c$ is an arbitrary act. By Theorem~\ref{thm:cancellationchar1}, rational mixture independence holds, so we have $ra + (1-r)c \succeq rb + (1-r)c$ for all rational $r$.  By $\Afour$, we have $ra + (1-r)c \succeq rb+(1-r)c$ for all real~$r$.   Conversely, suppose that $ra + (1-r)c \succeq rb + (1-r)c$ for some real~$r$.  If $r$ is rational, it is immediate from rational mixture independence that $a \succeq b$.  If $r$ is not rational, choose a rational $r'$ such that $0 < r' < r$.  Then we can find a sequence of rational numbers $r_n$ such that $r_n r$ converges to $r'$.  By rational mixture independence, $r_n (ra + (1-r)c)+(1-r_n) c \succeq r_n (rb + (1-r)c) + (1-r_n) c$.  By $\Afour$, it follows that $r'a + (1-r') c \succeq r' b + (1-r') c$.  Now by rational mixture independence, we have $a \succeq b$, as desired.

For part (b), suppose that $a \succ b \succ c$.  Mixture independence
(which follows from $\AthreeprpAX$ and $\Afour$, as we have observed)
implies that, for all $r\in (0,1)$, $a \succ ra + (1-r)c$.  To see that
$r a+(1-r)c \succ b$ for some $r \in (0,1)$, suppose not.  Then, by
$\Aone$, $b \succeq r a+(1-r)c$ for all $r \in (0,1)$, and by
$\Afour$, we have that $b \succeq a$, contradicting our initial assumption.  The remaining inequalities follow in a similar fashion.
\end{proof}

We now prove the representation theorems: Theorems~\ref{thm:eurepa}, \ref{thm:aaeurep}, and \ref{thm:outcomes}.   They all use essentially the same technique.  It is convenient to start with Theorem~\ref{thm:aaeurep}.  The first step is to get an additively separable utility representation
for AA acts on a state space $S$ with outcome space $O$.  This result is presented 
in Theorem \ref{thm:rep4} which is
somewhat novel because we use extended mixture cancellation and \Afour\
rather than independence and \Arch, and because $\succeq$ can be
incomplete.

\begin{theorem}\label{thm:rep4}
A preference relation $\succeq$ on a set $C$ of mixture-closed
AA acts
mapping a finite set $S$ of states
to distributions over a finite set $O$ of outcomes satisfies Extended
Mixture Cancellation and \Afour\
iff there
exists a
set $\mathcal{U}$ of utility functions on $S\times O$ such that
$a\succeq b$ iff
\begin{equation}\label{eq:char1}
\sum_{s\in S}\sum_{o\in O} u(s,o) a(s)(o)\geq
\sum_{s\in S}\sum_{o\in O} u(s,o) b(s)(o)
\end{equation} for all $u\in\U$.
Moreover, $\succeq$ also satisfies \Aone\ iff
we can take $\U$ to be a singleton $\{u\}$.  In this case, $u$ is unique up
affine transformations: if $u'$ also satisfies (\ref{eq:char1}),
then there exist $\alpha$ and $\beta$ such that $u' =
\alpha u + \beta$.
\end{theorem}

\begin{proof}
In the totally ordered case, this result is well known.  Indeed, for a
preference relation $\succeq$ that satisfies $\Aone$,
Proposition~7.4 of \cite{Kreps} shows that such a representation holds
iff $\succeq$ satisfies $\Arch$, mixture independence, transitivity, and reflexivity.   Theorem~\ref{thm:cancellationchar1} and
Proposition~\ref{prop:arch} show that if $\succeq$ satisfies extended
mixture cancellation, $\Aone$, and $\Afour$, then these properties hold, so
there is a  
representation.  Conversely, if there is such a representation, then
all these properties are easily seen to hold.  It follows that in the
presence of $\Aone$, extended mixture cancellation and $\Afour$ are equivalent
to these properties.   However, since we do not want to assume $\Aone$, we must work a little harder.  Fortunately, the techniques we use will be useful for our later results.

To see that the existence of a representation  implies Extended Mixture Cancellation and \Afour, first consider Extended Mixture Cancellation, and suppose that $\langle
a_1,\ldots,a_n\rangle$ and $\langle b_1\ldots,b_n\rangle$ are such that
$a_1 \succeq b_1$,
\ldots, $a_k \succeq b_k$, $a_{k+1}=\ldots=a_n$,
$b_{k+1}=\cdots=b_n$, and $a_1 + \cdots + a_n = b_1 + \cdots + b_n$.
For all $u \in \U$, for $i=1,\ldots,k$, we have
\begin{equation*}
\sum_{s\in S} u(s,a_i(s))\geq\sum_{s\in S} u(s,b_i(s)).
\end{equation*}
Since $a_1 + \cdots + a_n = b_1 + \cdots + b_n$, for
all $s\in S$, it must be that, for all $u\in\U$,
\begin{equation*}
\sum_{i=1}^n\sum_{s\in S}u(s,a_i(s))=\sum_{i=1}^n\sum_{s \in S}
u(s,b_i(s)).
\end{equation*}
Thus, for all $u\in \U$,
\begin{equation*}
\sum_{i=k+1}^n\sum_{s\in S}u(s,a_i(s))\leq\sum_{i=k+1}^n\sum_{s\in S}u(s,b_i(s)).
\end{equation*}
Since $a_{k+1}=\ldots=a_n$ and $b_{k+1}=\ldots=b_n$, it easily follows
that, for all $u\in U$,
\begin{equation*}
\sum_{s \in S} u(s,a_n(s))\leq\sum_{s \in S} u(s,b_n(s)).
\end{equation*}
Thus $b_n\succeq a_n$, as desired.
The fact that \Afour\ holds is straightforward, and left to the reader.

For the `if' direction, recall that we can view the elements of $C$ as
vectors in $\Nspace{|S| \times |O|}$.  For the rest of this proof, we
identify elements of $C$ with such vectors.  Let $D = \{a-b: a,b \in
C\}$, and let $D^+ = \{a -b: a \succeq b\}$.  Recall that a
\emph{(pointed) cone} in $\Nspace{|S| \times |O|}$ is a set $\CC$ that is
closed under nonnegative linear combinations, so that if $c_1, c_2 \in
\CC$  and $\alpha, \beta \ge 0$, then $\alpha c_1 + \beta c_2 \in CC$.
We need the following lemma.

\begin{lemma}\label{lem:cone} There exists a closed convex cone $\CC$
such that $D^+ =
\CC \cap D$.
\end{lemma}

\begin{proof}
Let $\CC$ consist of all vectors of the form $\alpha_1 d_1 + \cdots +
\alpha_n d_n$ for some $n > 0$, where $d_1, \ldots, d_n \in D^+$ and $\alpha_1,
\ldots, \alpha_n \ge 0$.  Clearly $\CC$ is a convex cone, and closed because $D^+$ is finite.   Also, $D^+ \subseteq \CC
\cap D$.  For the opposite inclusion, suppose that $\alpha_1 d_1 +
\cdots + \alpha_n d_n = d$, where $d_1, \ldots, d_n \in D^+$, $d \in D$,
and $\alpha_1, \ldots, \alpha_n \ge 0$.
Thus, there must exist
 $a_1, \ldots, a_n, b_1, \ldots, b_n, a, b \in C$ such
that $a-b = d$, $a_i - b_i = d_i$, and $a_i \succeq b_i$ for $i =
1,\ldots, n$. We want to show that $d \in D^+$ or, equivalently, that $a
\succeq b$.  Let $r = \alpha_1 + \cdots + \alpha_n + 1$.
Since $C$ is mixture-closed, both $(\alpha_1/r) a_1 + \cdots + (\alpha_n/r)
a_n + (1/r) b \in C$ and
$(\alpha_1/r) b_1 + \cdots + (\alpha_n/r) b_n + (1/r) a$ are in  $C$.
Moreover, since
$\alpha_1 d_1 + \cdots + \alpha_n d_n = d$ and $a_i \succeq b_i$ for $i =
1,\ldots, n$, it easily follows from mixture independence (which is a
consequence of extended mixture cancellation) that
$(\alpha_1/r) b_1 + \cdots + (\alpha_n/r) b_n + (1/r) a =
(\alpha_1/r) a_1 + \cdots + (\alpha_n/r) a_n + (1/r) b \succeq
(1/r) b_1 + \cdots + (1/r) b_n + (1/r) b.$
Another application of mixture independence gives us $a\succeq b$, as desired.
\end{proof}

Returning to the proof of Theorem~\ref{thm:rep4}, note that
if $a \succeq b$ for all $a$ and $b$, then $D^+ = D$, and we can take
$\CC$ to be the whole space.  If $\succeq$ is nontrivial, then $\CC$ is
not the whole space.  It is well known \cite{Rockafellar}
that every closed
cone that is
not the whole space is the intersection of closed half-spaces (where a
half-space is characterized by a vector $u$ such and consists of
all the vectors $x$ such that $u\cdot x \ge 0$).  Given our
identification of elements of $C$ with vectors, we can identify
the vector $u$ in $\Nspace{|S| \times |O|}$ characterizing a half-space
with a (state-dependent) utility function, where $u(s,o)$ is the $(s,o)$
component of the vector $u$.  If $\CC$ is the whole space, we can get a
representation by simply taking $\U$ to consist of the single utility
function such that $u(s,o) = 0$ for all $(s,o) \in S \times O$.
Otherwise, we can take $\U$ to consist of the utility functions
characterizing the half-spaces containing $\CC$.  It is easy to see that
for $a, b \in C$, we have that
$a \succeq b$ iff $a - b \in D^+$ iff $a - b \in \CC$  iff
$u \cdot (a-b) \ge 0$ for every half-space $u$ containing
$\CC$; i.e. iff (\ref{eq:char1}) holds.
\end{proof}

To prove Theorem~\ref{thm:aaeurep}, the following lemma, which shows that
we can identify complete preference relations with half-spaces, is also
useful. Given a subset $R$ of $\Nspace{|S| \times |O|}$, define the
relation $\succeq_R$ on $C$ by taking $a \succeq_R b$ iff $a-b \in R$.

\begin{lemma}\label{lem:halfspace}
$\EXAX = \{\succeq_R: R \mbox{ is either a half-space containing $\CC$
or the full space}\}$.
\end{lemma}

\begin{proof} If $R$ is the full space, then $\succeq_R$ is the trivial
relation, so clearly $\succeq_R \, \in \EXAX(\succeq)$.  If $R$ is a
half-space $H$ containing $\CC$ and $H$ is characterized by $u$, then
$\succeq_H$ extends $\succeq$, since $\CC \subseteq H$.
To see that $\succeq_H$ satisfies $\Aone$, observe that if $(a,b)
\notin \, \succeq_H$, then $u \cdot (a-b) < 0$, so $u \cdot
(b-a) > 0$, and $b \succeq_H a$.   To see that $\succeq_H$ satisfies
$\AthreeprAX$, suppose that $a_1 + \cdot + a_n = b_1 + \cdot + b_n$,
and $a_i \succeq_H b_i$ for $i = 1, \ldots, n-1$.  Thus,
$u \cdot (a_1 + \cdot + a_n ) = u \cdot (b_1 + \cdot b_n ) $,
and $u \cdot (a_i - b_i) \ge 0$ for $i = 1,\ldots, n-1$.  It follows
that $(b_n - a_n) \cdot u \ge 0$, so $b_n \succeq_H a_n$.
Finally, for \Afour, it is clear that if $(a_n,b_n) \rightarrow (a,b)$,
and $u \cdot (a_n - b_n) \ge 0$, then
$u \cdot (a - b) \ge 0$, so $a \succeq_H b$.
Thus, $\succeq_H \, \in \EXAX(\succeq)$.

For the opposite inclusion, suppose that $\succeq' \, \in
\EXAX(\succeq)$.  Let $\CC'$ be the cone determined by $\succeq'$, as in
Lemma~\ref{lem:cone}.  Clearly $\CC \subseteq \CC'$.  If $\CC'$ is the
full space, then we are done, since $\succeq' = \succeq_{\CC'}$.
Otherwise, $\CC'$ is the intersection of half-spaces.
Choose a half-space $H$ such that $\CC' \subseteq H$.  We claim that
$\succeq' = \succeq_{H}$.  Suppose not.  Since $\CC' \subseteq H$, we
must have $\succeq' \, \subseteq \, \succeq_{H}$.   There must exist
$a, b \in C$ such that $a \succeq_{H} b$ and $a {\not{\succeq}}' b$.
Since $\succeq'$ is complete, we must have $b \succ' a$.  Thus,
$b \succeq_{H} a$, so $a \sim_H b$.  Since $H$ is not the full space,
there must be some $c$ such that $b \, {\not\sim}_H c$.  Suppose
that $c \succ_H b$.  We must have $c \succ b$, since otherwise $b
\succeq' c$, and it follows that $b \sim _H c$.  By the Archimedean
property (which holds by Proposition~\ref{prop:arch}), since $c \succ' b
\succ' a$, there exists $r > 0$ such that $b \succ' rc + (1-r)a$.
Thus we must have $b \succeq_H rc + (1-r)a \succeq_H b$.  But this
contradicts the  assumption that $c \succ_H b \sim_H a$.  We get a
similar contradiction if $b \succ_H c$, since then $a \succ_H c$.
\end{proof}

\begin{proof}[\textbf{\upshape Proof of Theorems \ref{thm:aaeurep} and \ref{thm:aaeurep}c}]
It is easy to check that if there is a constructive SEU representation
of $\succeq$, then $\succeq$ satisfies $\AthreeprpAX$ and $\Afour$.
For the converse, suppose that $\succeq$ satisfies $\AthreeprpAX$ and
$\Afour$.  Take  $S=\AtomsAX(T_0)$ and $O' =\A_0$.  Define
$\pi_{S}(t)$ to be the set of all atoms $\delta$ in $\AtomsAX(T_0)$ such
that $\delta \rimp t$.  Define $\rho^0_{SO'}(a)$ to be the constant
function $a$ for a primitive choice $a$.  It is easy to see that
$\rho_{SO'}(a) = f_a$ for all choices $a \in C$.  Define a preference
relation $\succeq_S$ on the AA acts of the form $\rho_{SO'}(a)$ by taking
$f_a \succeq_S f_b$ iff $a\succeq b$.
The fact that $\succeq_S$ is well defined follows from
Proposition~\ref{prop:equivalence},  for if $f_a = f_{a'}$, then it
easily follows that $a \equiv_{\T} a'$, so $a \sim a'$.
Clearly $\succeq_S$ satisfies \ecsa, and satisfies \Aone\ iff $\succeq$ does.
Thus, by Theorem~\ref{thm:rep4}, there is an additively separable
representation of $\succeq$.

Now we adjust the state and outcome spaces to get a constructive SEU
representation.  Suppose first that \Aone\ holds.  Take
$S=\AtomsAX(T_0)$ and take $O=\AtomsAX(T_0)\times\A_0$.  For a primitive
choice $a\in\A_0$, define $\rho^0_{SO}(a)(\delta)=(\delta,a)$.  To
complete the proof, it clearly suffices to find a probability measure
$p$ on $\AtomsAX(T_0)$ and a utility function $v$ on
$\AtomsAX(T_0)\times\A_0$ such that $u(\delta,a)=p(\delta )v(\delta,a)$,
where
$u$ is the state-dependent utility function whose existence is guaranteed
by Theorem~\ref{thm:rep4}.  This is accomplished by taking
\emph{any} probability measure $p$ on $\AtomsAX(T_0)$ such that for all
atoms $\delta$, $p(\delta)>0$, and taking
$v(\delta,a)=u(\delta,a)/p(\delta)$.

If \Aone\ does not hold, then proceed as above, using Theorem
\ref{thm:rep4} to get an entire set $\mathcal{U'}$ of functions
$u:\AtomsAX(T_0) \times \A_0 \to\R$, and a single probability
distribution $p$ that assigns positive probability to every atom.  Let
$\mathcal{U}$ consist
of all utility functions $u$ such that there exists some $u' \in
\mathcal{U}'$ such that $u(\alpha,\delta)=
u'(\delta,a)/p(\delta)$.
In this representation, again, the state space is $\AtomsAX(T_0)$.

We now give a representation using a single utility function.
Let $S' = \AtomsAX(T_0)\times\EXAX(\succeq)$ and let $O'' = \AtomsAX(T_0)\times\A_0 \times \EXAX(\succeq)$.	
Define \mbox{$\rho^0_{S'O''}(a)(\delta,\succeq') = (\delta,a,\succeq')$.}
For $t\in T_0$, define $\pi_{S'}(t) = \pi_S(t)
\times\EXAX(\succeq)$.  As before, let $\U'$ be the set of utility
functions on $\AtomsAX(T_0) \times \A_0$ that represent
\mbox{$\succeq$}.  Lemma~\ref{lem:halfspace} shows that $\U'$
consists of one utility function $u_{\succeq'}$ for every total order
$\succeq' \in \EXAX(\succeq)$.  Again,
fix a
probability measure $p$ on $\AtomsAX(T_0)$ such that $p(\delta) > 0$
for all $\delta \in \AtomsAX(T_0)$.   For each relation $\succeq' \in
\EXAX(\succeq)$,
define $p_{\succeq'}$ on $\AtomsAX(T_0) \times
\EXAX(\succeq)$
by taking $p_{\succeq'}(\delta,\succeq'') = p(\delta)$ if $\succeq' =
\succeq''$, and $p_{\succeq'}(\delta,\succeq'') = 0$ if $\succeq' \ne
\succeq''$.
Let $\P = \{p_{\succeq'}: \succeq' \in \EXAX(\succeq)\}$.
Define $v(\delta,a,\succeq') =
u_{\succeq'}(\delta,a)/p(\delta)$.
It is easy  to see that $\P \times \{v\}$ represents $\succeq$.
Moreover, it easily follows that, with this representation, $\succeq_t$
and $\succeq \oplus\, (a,b)$ can be represented by updating.
\end{proof}

Now we show how the ideas in this proof can be modified to prove
Theorems~\ref{thm:eurepa} and~\ref{thm:eurepa}c.

\begin{proof}[\textbf{\upshape Proof of Theorems~\ref{thm:eurepa} and
\ref{thm:eurepa}c}]
Again, it is easy to check that if there is a constructive SEU
representation of $\succeq$, then $\succeq$ satisfies $\AthreeprAX$.

For the converse, suppose that we are given a preference relation $\succeq$ that satisfies
$\AthreeprAX$.  The structure of the proof is identical to that of
Theorem~\ref{thm:aaeurep}.  We first prove an analogue of
Theorem~\ref{thm:rep4}.

\begin{theorem}\label{thm:staterep}
A preference relation $\succeq$ on a  set $C$ of Savage acts
mapping a finite set $S$ of states
to a finite set $O$ of outcomes satisfies \ecsa\
iff there
exists a
set $\mathcal{U}$ of utility functions on $S\times O$ such that
$a\succeq b$ iff
\begin{equation}
  \sum_{s\in S}u\bigl(s,a(s)\bigr)\geq \sum_{s\in S}u\bigl(s,b(s)\bigr)
  \mbox{for all $u\in\mathcal{U}$.}
\end{equation}
Moreover, $\succeq$ satisfies \Aone\ iff $\U$ can be
chosen to be a singleton.%
\footnote{This result is a generalization of Theorem 4.1 of
  \citet{Fishburn70}, which can be viewed as considering
  the case where $|S| = 1$ and the preference relation $\succeq$ is
  complete (i.e., satisfies $\Aone$).}
\end{theorem}

The proof of Theorem~\ref{thm:staterep} is identical to that of
Theorem~\ref{thm:rep4}, except that
we need an analogue of Lemma~\ref{lem:cone} for the case that $\succeq$
satisfies $\AthreeprAX$.   Let $D$ and $D^+$ be defined as in
Lemma~\ref{lem:cone}.

\begin{lemma}\label{lem:cone1} If $\succeq$ satisfies
  $\AthreeprAX$, there exists a cone $\CC$ such that $D^+ =\CC \cap D$.
\end{lemma}

\begin{proof}
Again, let $\CC$ consist of all vectors of the form $\alpha_1 d_1 + \cdots + \alpha_n d_n$ for some $n > 0$, where $d_1, \ldots, d_n \in D^+$ and $\alpha_1, \ldots, \alpha_n \ge 0$.  Clearly $\CC$ is a cone.  Since $C$ is closed and bounded, so is $D$, and \Afour\ implies that  $D^+$ is closed and bounded.  Therefore $CC$ is closed.  Furthermore, $D^+ \subseteq \CC
\cap D$.  For the converse inclusion, suppose that $\alpha_1 d_1 +
\cdots + \alpha_n d_n = d$, where $d_1, \ldots, d_n \in D^+$, $d \in D$,
and $\alpha_1, \ldots, \alpha_n \ge 0$.
That means that $\alpha_1, \ldots, \alpha_n$ is a nonnegative solution
to the system of equations $x_1 d_1 + \cdots + x_n d_n = d$.  Since all the
coefficients in these equations are rational (in fact, they are all 0,
1, and $-1$) there exists a nonnegative rational solution to this
system of equations.  It easily follows that there exist positive integers
$\beta_1,
\ldots, \beta_{n+1}$ such that
$\beta_1 d_1 + \cdots + \beta_n d_n = \beta_{n+1}d$.
By definition,
there must exist $a_1, \ldots, a_n, b_1, \ldots, b_n, a, b \in C$ such
that $a-b = d$, $a_i - b_i = d_i$, and $a_i \succeq b_i$ for $i =
1,\ldots, n$.
It follows that $\{\{a_1, \ldots, a_1, \ldots,
a_n, \ldots, a_n, b, \ldots, b\}\} = \{\{b_1, \ldots, b_1, \ldots,
a_n, \ldots, a_n, b, \ldots, b\}\}$, where $a_i$ occurs in the left-hand
multiset $\beta_i$ times and $b$ occurs $\beta_{n+1}$ times, and,
similarly, $b_i$ occurs in the right-hand side $\beta_i$ times and $a$
occurs $\beta_{n+1}$ times.  $\AthreeprAX$ now implies
that $a \succeq b$, so $d \in \CC$, as desired.
\end{proof}

The proof of Theorem~\ref{thm:staterep} is now identical to that of
Theorem~\ref{thm:rep4}.  Moreover, the proof of Theorem~\ref{thm:eurepa}
now follows from Theorem~\ref{thm:staterep} in exactly the same way that
the proof of Theorem~\ref{thm:aaeurep} follows from
Theorem~\ref{thm:rep4}.
\end{proof}

\begin{proof}[\textbf{\upshape Proof of Theorems \ref{thm:outcomes} and \ref{thm:outcomes}c}] First, let $S = \AtomsAX(T_0)$ and $O' = \A_0$.  As in the proof of Theorem~\ref{thm:aaeurep}, using Theorem~\ref{thm:rep4}, we can find an additively separable representation of $\succeq$; that is, we can find a set $\U$ of utility functions on
$S= \AtomsAX(T_0) \times \A_0$ that represent $\succeq$.
Let $o_1$ and $o_0$ denote the best and worst outcomes guaranteed to
exist by $\Afive$.  Note that it follows from $\Afive$ and $\Asix$ that $u(\delta,o_0) \le u(\delta,a) \le u(\delta,o_1)$ for all $(\delta,a)\in S$.  (For null atoms $\delta$, we must in fact have
$u(\delta,o_0) = u(\delta,a) = u(\delta,o_1)$ for all $a \in \A_0$).
Furthermore, note that we can replace $u$ by $u'$, where $u'(\delta,a) =
u(\delta,a) - u(\delta,o_0)$ for all $(\delta,a) \in S$ to get an
equivalent representation; thus, by appropriate scaling, we can assume
without loss of generality that, for all $u \in \U$, we have that
$u(\delta,o_0) = 0$ for all $\delta \in \AtomsAX(T_0$) (so $u(\delta,a)
\ge 0$ for all $(\delta,a) \in S$) and that
$\sum_{\delta' \in \AtomsAX(T_0)} u_{\succeq'}(\delta',o_1) = 1$.
Finally, note that it easily
follows from $\Afour$ that, for all $o \in O$, there exists a unique
$c_o \in [0,1]$ such that $o \sim c_o o_1 + (1-c_o)o_0$ (in fact,
$c_o = \inf\{c: co_1 + (1-c)o_0 \succeq o\}$).  Clearly $c_{o_1} = 1$
and $c_{o_0} = 0$.
By $\Asix$, it follows
that $o \sim_\delta c_o o_1 + (1-c_o)o_0$ for all atoms $\delta$.  Thus,
we must have that $u(\delta,o) = c_0u(\delta,o_1)$ for all atoms $\delta$
and all $u \in \U$.

We now construct a state-independent SEU representation using $O$ as the
outcome space.
Let $S' = \AtomsAX(T_0) \times \A_0 \times
\EXAXp(\succeq)$.
Define $\pi^0_{S'}$ by taking $\pi^0_S(t) = \union_{\delta \rimp t}
\{\delta \times  A_0 \times \EXAXp(\succeq)\}$, and
define $\rho_{SO}(a)((\delta,a',\succeq'))$ to be $a$
if $a \in O$; $o_1$ if $a \in A_0 - O$ and $a \succeq_{\delta}' a'$; and
$o_0$ otherwise.
Let $u'$ be defined by taking  $u'(o) = c_o$.
Finally, recall that we can take $\U = \{u_{\succeq'}: \succeq' \in
\EXAXp(\succeq)\}$, where $u_{\succeq'}$ represents $\succeq'$.
Let $p_{\succeq'}$ be defined so that
$p_{\succeq'}(\delta,a,\succeq'') = 0$ and, for all $a \in \A_0$,
$p_{\succeq'}(\{(\delta,a',\succeq'): a \succeq' a'\}) =
u_{\succeq'}(\delta,a).$
It is easy to check that a probability measure
$p_{\succeq'}$ can be defined so as to satisfy
this constraint.  In particular, note that $p_{\succeq'}(\{\delta\} \times
\A_0 \times \{\succeq'\}) = u_{\succeq'}(\delta,o_1)$.
For all $(\delta,a) \in S$, we have that
$$u_{\succeq'}(\delta,a) = \sum_{a': \, a \succeq' a'}
p_{\succeq'}(\delta,a',\succeq') =  \sum_{a' \in \A_0}
p_{\succeq'}(\delta,a',\succeq')u(\rho_{S'O}(a)(\delta,a',\succeq')).$$
It follows that $\P$ and $u$ represent $\succeq$, where $\P=
\{p_{\succeq'}: \succeq' \in \EXAXp(\succeq)\}$.
As usual, it is straightforward to verify that updating works appropriately.
\end{proof}

\bibliographystyle{ecca}
\bibliography{beh}
\end{document}